\documentclass[acmsmall,screen]{acmart}

\AtBeginDocument{%
  \providecommand\BibTeX{{%
    \normalfont B\kern-0.5em{\scshape i\kern-0.25em b}\kern-0.8em\TeX}}}

\setcopyright{acmcopyright}
\copyrightyear{2021}
\acmYear{2021}
\acmDOI{10.1145/1122445.1122456}

\acmConference[OOPSLA '21]{OOPSLA '21}{}{OOPSLA}
\acmBooktitle{OOPSLA '21}
\acmPrice{15.00}
\acmISBN{978-1-4503-XXXX-X/18/06}

\graphicspath{ {./Figures/} }

\usepackage[inference]{semantic}
\usepackage{hyperref} 
\usepackage{xcolor}
\usepackage[]{algorithm}
\usepackage{algorithmic}
\usepackage{enumerate}
\usepackage{proof}

\newtheorem{theorem}{Theorem}[section]

\newtheorem{lemma}[theorem]{Lemma}
\newtheorem{prop}[theorem]{Property}

\theoremstyle{definition}
\newtheorem{definition}{Definition}[section]

\newcommand{\st}{\text{s.t. }}
\newcommand{\br}[1]{\langle a #1 b\rangle}

\newcommand{\callsym}[1]{\langle {#1}}
\newcommand{\retsym}[1]{{#1} \rangle}
\newcommand{\ac}{\callsym{a}}
\newcommand{\bc}{\retsym{b}}
\newcommand{\brp}[3]{\callsym{#1} #3 \retsym{#2}}

\newcommand{\la}{\leftarrow}
\newcommand{\ra}{\rightarrow}
\newcommand{\La}{\Leftarrow}
\newcommand{\Ra}{\Rightarrow}
\newcommand{\deriv}{\rightarrow^{*}}
\newcommand{\trans}{\delta}

\newcommand{\hd}{\text{head }}

\newcommand{\RTrans}{\mathcal{T}}

\newcommand{\tail}{\text{tail}}
\newcommand{\none}{\text{None}}

\newcommand{\mc}{m_{\text{call}}}
\newcommand{\mr}{m_{\text{ret}}}
\newcommand{\ml}{m_{\text{pln}}}
\newcommand{\mmc}{\mathcal{M}_{\text{call}}}
\newcommand{\mmr}{\mathcal{M}_{\text{ret}}}
\newcommand{\mml}{\mathcal{M}_{\text{pln}}}

\newcommand{\true}{\text{true}}
\newcommand{\false}{\text{false}}

\newcommand{\ptBigstepSym}{\Ra}
\newcommand{\ptBig}[3]{{#1} \ptBigstepSym \left({#2},{#3}\right)}
\newcommand{\ptSmallstepSym}{\rightarrow}
\newcommand{\ptSmall}[5]{({#1},{#2}) \stackrel{#3}{\ptSmallstepSym} ({#4},{#5})}
\newcommand{\ptSmallStar}[5]{({#1},{#2}) \stackrel{#3}{\longrightarrow^*} ({#4},{#5})}

\newcommand{\firstNT}[1]{\mathrm{firstNT}(#1)}
\newcommand{\lastNT}[1]{\mathrm{lastNT}(#1)}

\newcommand{\derivRule}[2]{{#1} \rightarrow {#2}}
\newcommand{\derivStar}[2]{{#1} \rightarrow^* {#2}}

\newcommand{\connect}{\diamond}
\newcommand{\parseInv}[1]{\mathrm{Inv}(#1)}

\newcommand{\pPDAtrans}{\mathcal{P}}
\newcommand{\gPDA}{\mathcal{T}_g}
\newcommand{\gPDAtrans}{\mathcal{G}}

\newcommand{\extractor}[1]{\mathrm{extract}(#1)}

\newcommand{\secondRound}[1]{{{#1}}}

\newcommand{\emptyL}{[]}
\newcommand{\singletonL}[1]{[{#1}]}
\newcommand{\consL}[2]{{#1}{::}{#2}}

\newcommand{\set}[1]{\{#1\}}
\newcommand{\predset}[2]{\{#1\,\mid\,#2\}}
\newcommand{\Pow}{\mathop{\mathcal{P}}}

\setcitestyle{round}

\setcopyright{acmcopyright}
\acmPrice{}
\acmDOI{10.1145/3485528}
\acmYear{2021}
\copyrightyear{2021}
\acmSubmissionID{oopsla21main-p245-p}
\acmJournal{PACMPL}
\acmVolume{5}
\acmNumber{OOPSLA}
\acmArticle{151}
\acmMonth{10}

\begin{document}

\title{A Derivative-based Parser Generator for Visibly Pushdown Grammars}

\author{Xiaodong Jia}
\affiliation{%
  \institution{The Pennsylvania State University}
  \streetaddress{201 Old Main}
  \city{State College}
  \state{Pennsylvania}
  \country{USA}
  \postcode{16802}
}

\author{Ashish Kumar}
\affiliation{%
  \institution{The Pennsylvania State University}
  \streetaddress{201 Old Main}
  \city{State College}
  \state{Pennsylvania}
  \country{USA}
  \postcode{16802}}

\author{Gang Tan}
\affiliation{%
  \institution{The Pennsylvania State University}
  \streetaddress{201 Old Main}
  \city{State College}
  \state{Pennsylvania}
  \country{USA}
  \postcode{16802}
}

\renewcommand{\shortauthors}{Xiaodong Jia, Ashish Kumar and Gang Tan}

\begin{abstract}
  In this paper, we present a derivative-based, functional recognizer
  and parser generator for visibly pushdown grammars.  The generated
  parser accepts ambiguous grammars and produces a parse forest
  containing all valid parse trees for an input string in linear time.
  Each parse tree in the forest can then be extracted also in linear
  time.  Besides the parser generator, to allow more flexible forms of
  the visibly pushdown grammars, we also present a translator that
  converts a tagged CFG to a visibly pushdown grammar in a sound way,
  and the parse trees of the tagged CFG are further produced by
  running the semantic actions embedded in the parse trees of the
  translated visibly pushdown grammar. The performance of the parser
  is compared with a popular parsing tool ANTLR
  \secondRound{and other popular hand-crafted parsers.}
  The correctness of the core parsing algorithm is
  formally verified in the proof assistant Coq.
\end{abstract}

\ccsdesc[500]{Software and its engineering~Parsers}
\ccsdesc[300]{Software and its engineering~Software verification}
\ccsdesc[500]{Theory of computation~Grammars and context-free languages}

\keywords{parser generators, formal verification, derivative-based parsing}


\maketitle

\section{Introduction}
Parsing is a fundamental component in computer systems. Modern parsers
used in high-performance settings such as web browsers and network
routers need to be efficient, as their performance is critical to the
performance of the whole system. Furthermore, high-assurance parsers
are becoming increasingly more important for security, in settings
such as web applications, where their parsers are directly processing
potentially adversarial inputs from the network. In these settings,
formally verified parsers are highly desirable.

Most parsing libraries are based on Context-Free Grammars (CFGs) or
their variants. Although very flexible, CFGs have limitations in terms
of efficiency and formal verification. First, not all CFGs can be
converted to deterministic pushdown automata (PDA); the inherent
nondeterminism in some CFGs causes the worst-case running time of
general CFG-based parsing algorithms to be $O(n^3)$. Moreover,
formally verifying general CFG parsing algorithms is a
difficult task. To our best knowledge, there is no formally verified
CFG parsing algorithm due to its complexity.

To achieve efficient parsing, many parsing frameworks place
restrictions on what CFGs can be accepted, at the expense of
placing the burden on users to refactor their grammars to satisfy
those restrictions. Please see the related-work section for discussion
about common kinds of restrictions, leading to parsing
frameworks such as LL(k), LR(k)~\cite{LRParsing}, PEGs (Parsing Expression Grammars)~\cite{Ford04PEG}, etc.

This paper explores an alternative angle of building parsers based on
Visibly Pushdown Grammars (VPGs)~\cite{AlurM09NWA}. In VPGs, users
explicitly partition all terminals into three kinds: plain, call,
and return symbols. This partitioning makes the conversion of a VPG
to a deterministic PDA always possible, which provides the foundation for
efficient algorithms. Compared to requiring users to refactor their grammars to satisfy
restrictions placed by parsing frameworks such as LL(k), asking users to
specify what nonterminals are call and return symbols is less of a burden.

VPGs have been used in program analysis, XML processing, and other
fields, but their potential in parsing has not been fully exploited.  In
this paper, we show that VPGs bring many benefits in parsing. First,
we show an efficient, linear-time parsing algorithm for VPGs. Second,
our algorithm is amenable to formal verification.  Overall, this paper
makes the following contributions.

\begin{itemize}

    \item We present a derivative-based algorithm for VPG recognition
      and parsing. The algorithm is guaranteed to run in linear
      time. The generated parser accepts ambiguous grammars and
      produces a parse forest for the input string, where each parse
      tree in the forest can be extracted in linear time.

    \item We mechanize the correctness proofs of the parsing algorithm
      in Coq.

    \item We present a surface grammar called tagged CFGs to allow
      more convenient use of our parsing framework. Users can use
      their familiar CFGs for developing grammars and provide
      additional tagging information on nonterminals. A sound
      translator then converts a tagged CFG to a VPG.  

\end{itemize}

The remainder of this paper is organized as follows.
We first introduce VPGs in Section~\ref{sec:background} and  
discuss related work in Section~\ref{sec:related_work}.
Section~\ref{sec:recognition} presents a derivative-based VPG recognizer, 
which enlightens the parsing algorithm discussed in Section~\ref{sec:parsing}. 
The translator and tagged CFGs are discussed in Section~\ref{sec:surface_grammar}.
We then evaluate the VPG parser in Section~\ref{sec:eval}.

\section{Background} \label{sec:background}
As a class of grammars, VPGs~\cite{AlurM09NWA} have been used in
program analysis, XML processing, and other fields. Compared with
CFGs, VPGs enjoy many good properties. It is always possible to build
a deterministic PDA from a VPG. The terminals in a VPG are partitioned
into three kinds and the stack action associated with an input symbol
is fully determined by the kind of the symbol: an action of pushing to
the stack is always performed for a \emph{call symbol}, an action of
popping from the stack is always performed for a \emph{return symbol},
and no stack action is performed for a \emph{plain
  symbol}. Furthermore, VPGs enjoy all closure properties, including
intersection and complement.  As will be shown in this paper, these
properties enable the building of linear-time parsers for VPGs, and
make VPGs amenable to formal verification.  The expressive power of VPG
is between regular grammars and CFGs, and is sufficient for describing
the syntax of many practical languages, such as JSON, XML, and HTML.

We next give a formal account of VFGs. A grammar $G$ is represented as a tuple
$(V, \Sigma, P, L_0)$, where $V$ is the set of nonterminals, $\Sigma$
is the set of terminals, $P$ is the set of production rules, and
$L_0\in V$ is the start symbol.  
\secondRound{The alphabet $\Sigma$ is partitioned
into three sets: $\Sigma_l$, $\Sigma_c$, $\Sigma_r$, which contain
plain, call and return symbols, respectively.  Notation-wise, a terminal in
$\Sigma_c$ is tagged with $\langle$ on the left, and a terminal in
$\Sigma_r$ is tagged with $\rangle$ on the right. For example, $\ac$
is a call symbol in $\Sigma_c$, and $\bc$ is a return symbol
in $\Sigma_r$.}

We first formally define \emph{well-matched VPGs}.  
Intuitively, a well-matched VPG generates
only well-matched strings, in which a call symbol is always matched
with a return symbol in a derived string.
\begin{definition}[Well-matched VPGs]
    A grammar $G = (V, \Sigma, P, L_0)$ is a well-matched VPG with respect to the partitioning $\Sigma=\Sigma_l\cup\Sigma_c\cup \Sigma_r$, if every production rule in $P$ is in one of the following forms.
    \begin{enumerate}
        \item $L\ra\epsilon$, where $\epsilon$ stands for the empty string;
        \item $L\ra c L_1$, where $c\in\Sigma_l$;
        \item $L\ra \br{L_1}L_2$, where $\ac\in\Sigma_c$ and $\bc\in\Sigma_r$.
    \end{enumerate}
\end{definition}
Note that in $L\ra c L_1$ terminal $c$ must be a plain
symbol, and in $L\ra \br{L_1}L_2$ a call symbol must be matched with a
return symbol; these requirements ensure that any derived string must
be well-matched.

\secondRound{The following is an example of a well-matched VPG, which is taken 
from a VPG for XML:
$$
\text{element} \ra \text{OpenTag content CloseTag}
   \mid \text{SingleTag}
$$
In this example, nonterminals start with a lowercase character,
such as ``element'', and terminals start with an uppercase character,
such as ``OpenTag''. The grammar shows a typical usage of VPGs to
model a \emph{hierarchically nested matching} structure of XML texts:
``OpenTag'' is matched with ``CloseTag'', and ``content'' nested in
between can be ``element'' itself (not shown in the above snippet) and forms
an inner hierarchy.}

In the rest of the paper, we use the term VPGs for well-matched VPGs and
use the term general VPGs to allow the case of \emph{pending calls and
  returns}, which means that a call/return symbol may not have its
corresponding return/call symbol in the input string. To accommodate
pending symbols, general VPGs, in addition, allow rules in the forms of
$L\ra \ac L'$ and $L\ra \bc L'$, which we call \emph{pending rules}. 
Further, the set of nonterminals $V$
is partitioned into $V^0$ and $V^1$: nonterminals in $V^0$ only
generate well-matched strings, while nonterminals in $V^1$ can
generate strings with pending symbols.

\begin{definition}[General VPGs]
\label{def:generalVPG}
    A grammar $G = (V, \Sigma, P, L_0)$ is a general
    VPG with respect to the partitioning $\Sigma=\Sigma_l\cup\Sigma_c\cup \Sigma_r$ and $V=V^0\cup V^1$, if every rule in $P$ is in one of the following forms:
    \begin{enumerate}
        \item $L\ra\epsilon$;
        \item $L\ra i L_1$, where $i \in \Sigma$, and if $L\in V^0$ then (1) $i\in\Sigma_l$ and (2) $L_1\in V^0$;
        \item $L\ra \br{L_1}L_2$, where $\ac \in\Sigma_c$, $\bc  \in \Sigma_r$, $L_1\in V^0$, and if $L\in V^0$, then $L_2\in V^0$.
    \end{enumerate}
\end{definition}
The above definition imposes constraints on how $V^0$ and $V^1$
nonterminals can be used in a rule. For example, in $L\ra
\br{L_1}L_2$, nonterminal $L_1$ must be a well-matched nonterminal; so
$P$ cannot include rules such as $L_1 \ra \ac L_3$, since $L_1$ is
supposed to generate only well-matched strings.

The notion of a derivation in VPGs is the same as the one in CFGs.
We write $w \ra w'$ to mean a single derivation step according to a grammar, where $w$
and $w'$ are strings of terminals or nonterminals. We write
$\derivStar{L}{w}$ to mean that $w$ can be derived from $L$ via a
sequence of derivation steps.

\section{Related work} \label{sec:related_work}
Most parser libraries rely on the formalism of Context-Free Grammars
(CFGs) and user-defined semantic actions for generating parse
trees. Many CFG-based parsing algorithms have been proposed in the
past, including LL(k), LR(k)~\cite{LRParsing},
Earley~\cite{EarleyParsing},
CYK~\cite{CockeCYKParsing,YoungerCYKParsing,KasamiCYKParsing}, among
many others.  
\secondRound{LL(k) and LR(k) algorithms are commonly used, but their input
grammars must be unambiguous. Users often have to change/refactor
their grammars to avoid conflicts in LL(k) and LR(K) parsing tables,
a nontrivial task. Earley, CYK, and GLR parsing can handle
any CFG, but their worst-case running time is $O(n^3)$. In contrast,
our VPG parsing accepts ambiguous grammars and is linear time.}

Our VPG parsing algorithm relies on \emph{derivatives}.  One major
benefit of working with derivatives is that it is amenable to formal
verification, as proofs related to derivatives involve algebraic
transformations on symbolic expressions; they are easier to develop in
proof assistants than it is to reason about graphs (required when
formalizing LL and LR algorithms).  \citet{Brzozowski64}
first presented the concept of derivatives and used it to build a
recognizer for regular expressions.  The idea was
revived by \citet{Owens09Derivatives} and generalized to generate
parsers for CFGs~\cite{Might11Derivatives}, with an exponential worst-case time
complexity.  More recently, a symbolic
approach~\cite{henriksen2019derivative} for parsing CFGs with
derivatives was presented, with cubic time complexity. Finally,
\citet{darragh2020parsing} presented a formally
verified, derivative-based, linear-time parsing algorithm for 
LL(1) context-free expressions.

Owl is an open-source
project\footnote{\url{https://github.com/ianh/owl}} that provides
a parser generator for VPGs. It has the same goal as our work,
but differs in the following critical aspects:
(1) Owl supports only well-matched VPGs, while our parsing library supports full VPGs;
(2) Owl adopts a different algorithm and is implemented in an imperative way, 
while  our parsing library is derivative-based and functional; 
(3) Owl does not provide formal assurance, while our parsing library is
formally verified in Coq;
(4) Owl rejects ambiguous VPGs, while our parsing library accepts ambiguous grammars 
and generates parse forests; and
(5) Owl does not support semantic actions embedded in grammars, while 
our parsing library accepts semantic actions.

Due to parsers' importance to security, many efforts have been made to
build secure and correct parsers. One obvious approach is testing,
through fuzz testing or differential testing
(e.g., \citet{Petsios17Nezha}).  However, testing cannot show the absence
of bugs. Formal verification has also been applied to the building of
high-assurance parsers. 
\citet{JourdanPL12ParserVal} applied the methodology of translation validation 
and implemented a verified parser validator for LR(1) grammars.
RockSalt~\cite{Morrisett12RockSalt} included a verified parser for
regular expression based DSL. 
\secondRound{\citet{lasser2019verified} and ~\citet{Edelmann20zippy} presented verified LL(1) parsers
but we are not aware of fully verified LL(k) parsers.
\citet{Lasser21costar} implemented a verified ALL(*) parser, which is the
algorithm behind ANTLR4.}
\citet{Koprowski10TRX}
implemented a formally verified parser generator for Parsing
Expression Grammars (PEG).
\citet{Ramananandro19EverParse} presented a verified parser generator
for tag-length-value binary message format descriptions.  We formalize
our derivative-based, VPG parsing algorithm and its correctness proofs in Coq.

\section{VPG based recognition} \label{sec:recognition}
We next present an algorithm for converting a VPG into a
deterministic PDA using a derivative-based
algorithm. The resulting PDA accepts the same set of strings as the input VPG.

Before we present the formal conversion process, we discuss informally
the intuition of the states, the stack, and the transition function of
the PDA that is converted from an input VPG.  A state in the resulting
PDA is a subset of $V \times V$, i.e., a set of nonterminal pairs.  A
nonterminal pair $(L_1, L_2)$ tells that the next part of the input
should match $L_2$ and the current context is $L_1$. The
\emph{context} is the nonterminal that is used to derive $L_2$,
without consuming an unmatched call symbol before getting to $L_2$.
Formally, it means there exists a derivation sequence
$\derivStar{{L}_1}{\omega_1 {L}_2 \omega_2}$, where $\omega_1$ is
a sequence of terminals and does not contain an unmatched call symbol,
and $\omega_2$ is a sequence of terminals or nonterminals. 

To give an example, suppose we have the following VPG rules, with
$L_0$ being the start symbol. We omit the rules for $L_2$ and
$L_3$, which are irrelevant for the discussion.
\[
L_0 \rightarrow c L_1; \ L_1 \rightarrow \ac L_2 \bc L_3.
\]
The start state of the PDA should be $\set{(L_0,L_0)}$, meaning that
the input string should match $L_0$ and the context is also $L_0$
since $L_0$ can be derived from itself (in zero steps) without
generating an unmatched call symbol. Given that start state, if the next input
symbol is $c$, then the PDA should transition to state $\set{(L_0,
  L_1)}$; that is, the rest of the input should match $L_1$ and the
context is still $L_0$, since $L_1$ is derived from $L_0$ without
generating an unmatched call symbol. Now suppose the next input char is $\ac$;
then the next state should be $\set{(L_2, L_2)}$; notice that there is
a context switch as there is an unmatched call symbol that is encountered when
going from $L_1$ to $L_2$ using the rule $L_1 \rightarrow \ac L_2 \bc
L_3$.  As we will show, for this transition, the PDA will also push
$\set{(L_0,L_1)}$ and $\ac$ to the stack, so that when the return symbol
$\bc$ is encountered, we can use that stack information to look up the old
context and transition the PDA to state $\set{(L_0,L_3)}$.

For this example, all states contain just one pair. In general,
a state may contain multiple pairs because of possible
ambiguity. For example, imagine there is an additional rule $L_0 \rightarrow
c L_4$; then from start state $\set{(L_0,L_0)}$, after encountering
$c$, the PDA transitions to state $\set{(L_0,L_3), (L_0,L_4)}$,
reflecting that the rest of the input can match either $L_3$ or $L_4$.

Given the above discussion, we have the following PDA states and stacks.
\begin{definition}[PDA states and stacks] 
    Given a VPG, we introduce a PDA
    whose states are subsets of $V \times V$ and whose stack contains
    stack symbols of the form $[S, \ac]$, where $S$ is a PDA state and
    $\ac \in\Sigma_c$ a call symbol. We write $\bot$ for the empty stack,
    and $[S,\ac ]\cdot T$ for a stack whose top is $[S,\ac ]$ and the rest
    is $T$. Intuitively, the stack remembers
    a series of past contexts, which are used for matching future
    return symbols. We call a pair $(S,T)$ a \emph{configuration},
    with $S$ being the state and $T$ being the stack.
\end{definition}

\secondRound{We next utilize the notion of
derivatives~\cite{Brzozowski64,Owens09Derivatives,Might11Derivatives}
to compute a recognizer PDA that accepts the same language (i.e., 
a set of strings) as an
input VPG. The
derivative of a language $\mathcal{L}$ with respect to an input symbol $c$ is the
residual set of strings of those in $\mathcal{L}$ that start with $c$:
\[
\delta_c(\mathcal{L}) = \predset{w}{cw \in \mathcal{L}}
\]

As we will show in Definition~\ref{def:semOfVPAStats}, a recognizer PDA
configuration $(S,T)$ stands for a language. Transferring the general
definition of derivatives to recognizer PDA configurations
produces a set of derivative functions, discussed next.}

Given a VPG $G=(V,\Sigma,P,L_0)$, we define three kinds of derivative
functions: (1) $\delta_c$ is for when the next input symbol is a plain
symbol $c$; (2) $\delta_{\ac}$ for when the next input symbol is a call
symbol $\ac$; and (3) $\delta_{\bc}$ for when the next input symbol is a return
symbol $\bc$. Each function takes the current state $S$ and the top
stack symbol, and returns a new state as well as an action on the
stack (expressed as a lambda function).  Note that $\delta_c$ and
$\delta_{\ac}$ do not need information from the stack;
therefore we omit the top stack symbol from their parameters.

\begin{figure}
    \centering
    \includegraphics[width=0.4\textwidth]{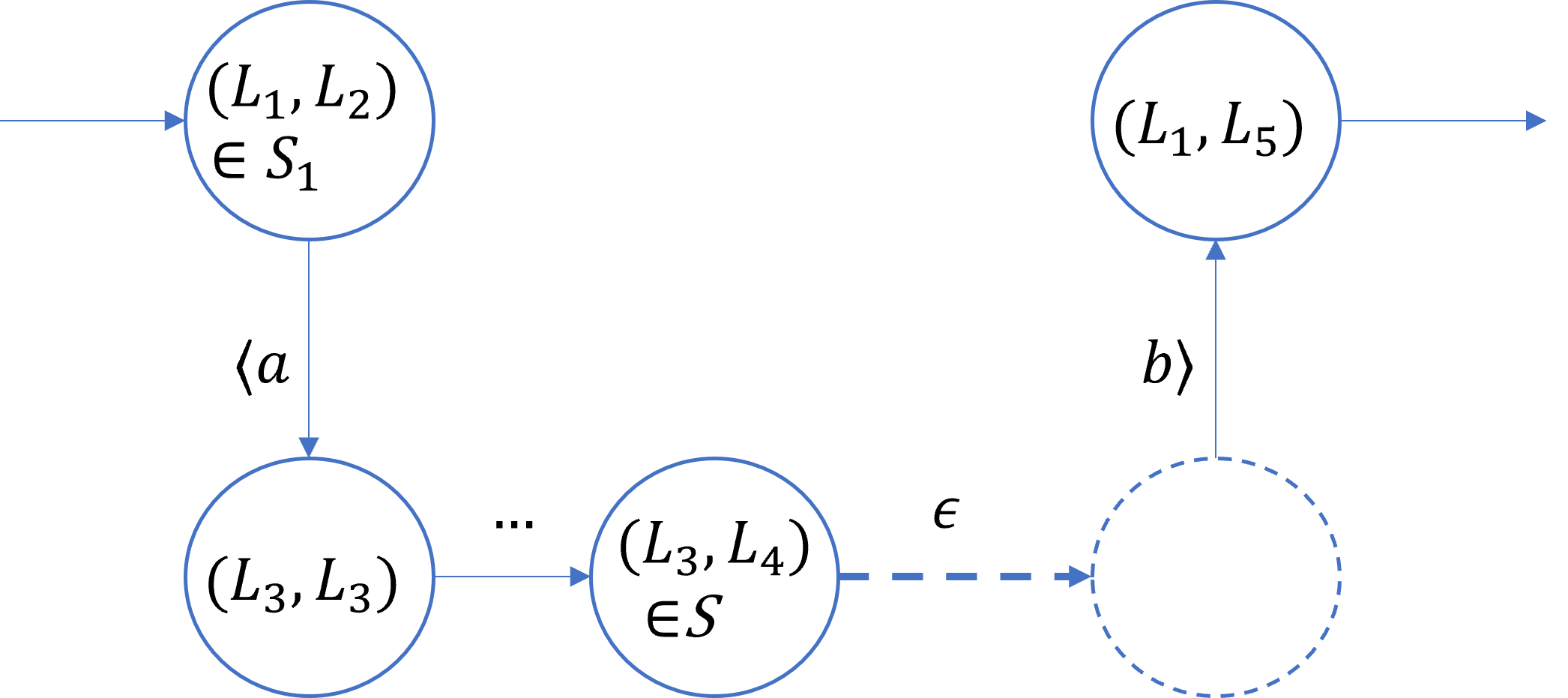}
    \caption{An example of state transitions: $(L_1,L_2)$ transfers to $(L_3,L_3)$ 
        with symbol $\ac$ and there is a rule $L_2\ra\br{L_3}L_5$, 
        and finally transfers to $(L_1,L_5)$ with symbol $\bc$.}
    \Description{An example of state transition}
    \label{fig:deltaRet}
\end{figure}

\begin{definition}[Derivative functions]
    \begin{enumerate}[(1)]
        \item[]
        \item 
            $\delta_c(S) = (S',\lambda T.T)$, where
                $$
                  S'=\{(L_1,L_3)\mid \exists L_2, (L_1,L_2)\in S \land (L_2\ra cL_3)\in P\};
                $$
         For a plain symbol $c\in\Sigma_l$, it checks each pair
         $(L_1,L_2)$ in the current state $S$, and if there is a rule
         $L_2\ra c\ L_3$, pair $(L_1,L_3)$ becomes part of the new
         state. In addition, the stack is left unchanged.
        \item
              $\delta_{\ac}(S) = (S',\lambda T.[S,\ac ]\cdot T)$, where
                $$
                  S'=\{(L_3,L_3)\mid \exists L_1\; L_2, (L_1,L_2)\in S\ \land
                     \exists L_4, (L_2\ra \brp{a}{b}{L_3}L_4)\in P \};
                $$
         For a call symbol $\ac\in\Sigma_c$, it checks each pair
         $(L_1,L_2)$ in the current state $S$; if there is a rule $L_2
         \ra \br{L_3} L_4$, pair $(L_3,L_3)$ becomes part of the new
         state; note there is a context change since a call symbol is
         encountered. In addition, the old state together with $\ac$
         is pushed to the stack.
        \item
             $\delta_{\bc}(S,[S_1,\ac]) =
                  (S',\tail)$,
              where
                $$
                S'=
                \{(L_1,L_5)\mid \exists L_2\; L_3\; L_4, (L_1,L_2)\in S_1 \land
                (L_3,L_4)\in S \land \\
                (L_4\ra\epsilon)\in P\ \land 
                (L_2\ra\brp{a}{b}{L_3}L_5)\in P \}.
                $$
              For a return symbol $\bc \in \Sigma_r$ and a stack top
              $[S_1,\ac]$, it checks each pair $(L_1,L_2)$ in the state
              $S_1$ of the stack top symbol; if there is a pair $(L_3,L_4)$ in the
              current state $S$, $L_4$ can derive the empty string, and
              there is a rule $L_2\ra \br{L_3}\ L_5$, then pair $(L_1,L_5)$
              becomes part of the new state; \secondRound{note that it checks
              $L_4\ra \epsilon$ to ensure that the current level is finished
              before returning to the upper level.} In addition, the stack top is popped from
              the stack.  Figure~\ref{fig:deltaRet} presents a drawing that
              depicts the situation when a return symbol is encountered.
    \end{enumerate}
\end{definition}

We formalize the semantics of PDA configurations as sets of accepted strings:
\begin{definition}[Semantics of PDA configurations]
\label{def:semOfVPAStats}
    We write $(S, T)\leadsto w$ to mean that a terminal string $w$ can be
    accepted by the configuration $(S, T)$. It is defined as follows.
    \begin{enumerate}
        \item $(S, \bot)\leadsto w \mbox{ if } \exists (L_1,L_2)\in S,\ \st L_2\deriv w$,
        \item $(S, [S',\ac ]\cdot T')\leadsto w_1 \bc  w_2$ if
              $\exists (L_3,L_4)\in S$ s.t.
              \begin{enumerate}
                  \item $L_4\deriv  w_1$ and
                  \item $\exists (L_1,L_2)\in S',\exists L_5, (L_2\ra \br{L_3}L_5) \in P \land (\{(L_1,L_5)\},\ T')\leadsto w_2$.
              \end{enumerate}
    \end{enumerate}
\end{definition}

The correctness of derivative functions is
stated in the following theorem, whose correctness proof is detailed in Appendix~\ref{app:correct_recognizer}.
Take the case of $\delta_c$ as an example: the theorem states that 
$(S,T)$ matches $cw$ iff
the configuration after running $\delta_c$ matches $w$ (the string after consuming $c$).
\begin{theorem}[Derivative function correctness]
\label{derivCorrect}
~~~~~~~~~
\begin{itemize}
\item    Assume $\delta_c(S)=(S',\lambda T.T)$ for a plain symbol $c$. Then
    $(S,T)\leadsto cw$ iff $(S',T) \leadsto w$.
\item   Assume $\delta_{\ac}(S)=(S',\lambda T.[S,\ac ]\cdot T)$ for a call symbol $\ac$.
    Then $(S,T)\leadsto \ac w$ iff $(S', [S,\ac ]\cdot T)\leadsto w$.
\item     Assume $\delta_{\bc}(S, [S_1,\ac ])=(S',\tail)$ for a return symbol $\bc$.
    Then $(S, [S_1,\ac ]\cdot T)\leadsto \bc  w$ iff
    $(S',T)\leadsto w$.
\end{itemize}
\end{theorem}

With those derivative functions, we can convert a VPA to a PDA, whose
set of states is the least solution of the following
equation; it makes sure that states are closed under derivatives.
\[
\begin{array}{lll}
A & = &
     A \cup~  \predset{S'}{c\in\Sigma_l,\ S\in A, \ \delta_c(S)=(S',f)}  \\
  &  & \cup~ \predset{S'}{\ac\in\Sigma_c,\ S\in A,\ \delta_{\ac}(S)=(S',f)} \\
  &  & \cup~ \predset{S'}{\bc\in\Sigma_r,\ \ac\in\Sigma_c,\ S \in A,\ S'\in A,\ \delta_{\bc}(S,[S',\ac])=(S',f)} \\
\end{array}
\]

\secondRound{We note that the least solution to the previous equation may
  include unreachable states, since the last line of the equation
  considers all $(S, [S', \ac])$ without regard for whether
  such a configuration is possible. This may make the resulting PDA
  contain more states and occupy more space for the PDA representation
  than necessary. However, unreachable states do not affect the
  linear-time parsing guarantee of VPG parsing, as during parsing
  those unreachable states are not traversed; further, during
  experiments we did not experience space issues when representing PDA
  states and transitions.}

Algorithm~\ref{alg:constr_recog} is an iteration-based method to solve
the equation for the least solution, where the returned $S_0$ is the initial state, $A$ is
the set of all states, and $\RTrans$ is the set of edges between
states. For an iteration, $N$ is the set of states that the
algorithm should perform derivatives on. 
\secondRound{Line 7 then performs derivatives using call and plain symbols and 
line 10 performs derivatives using return symbols.}

At the end of each iteration,
the following invariants are maintained: (1) $N \subseteq A$; (2) for
state $S \in A-N$ and $i \in \Sigma_c \cup \Sigma_l$, if
$\delta_i(S)=(S',f)$, then $S' \in A$; (3) for states $S, S' \in A-N$,
$\ac \in \Sigma_c$, and $ \bc \in \Sigma_r$, if $\delta_{\bc}(S,
[S',\ac])=(S'',f)$, then $S'' \in A$.  With these invariants, when $N$
becomes empty, $A$ is closed under derivatives.


\begin{algorithm}[t]
    \caption{Constructing the recognizer PDA}
    \label{alg:constr_recog}
    \begin{algorithmic}[1]
        \STATE Input: a VPG $G=(V,\Sigma,P,L_0)$, $\trans$.
        \STATE $S_0\la \{(L_0,L_0)\}$.
        \STATE Initialize the set for new states $N=\{S_0\}$.
        \STATE Initialize the set for all produced states $A=N$.
        \STATE Initialize the set for transitions $\RTrans=\{\}$.
        \REPEAT
        \STATE $N'\la \{ (i,f,S,S')\mid (S',f)=\delta_{i}S,\ S\in N, \mbox{ and } i\in\Sigma_c\cup\Sigma_l \}$
        \STATE {Add edge $(S,S')$ marked with $(i,f)$ to $\RTrans$, where $(i,f,S,S')\in N'$}.
        \STATE $R \la \{ [S,\ac]\mid  S\in A \mbox{ and } \ac\in\Sigma_c  \}$
        \STATE {$N_R\la\{ (\bc,r,f,S,S')\mid (S',f)=\delta_{\bc}(S,r), S\in A, \bc\in\Sigma_r, r\in R \}$}
        \STATE {Add edge $(S,S')$ marked with $(\bc,r,f)$ to $\RTrans$, where $(\bc,r,f,S,S')\in N_R$}.
        \STATE {$N\la\{S'\mid (\_,\_,\_,S')\in N' \lor  (\_,\_,\_,\_,S')\in N_R\} -  A$}
        \STATE $A\la A\cup N$
        \UNTIL{$N=\emptyset$}
        \STATE Return $(S_0,A,\RTrans)$.
    \end{algorithmic}
\end{algorithm}

As an example, consider the VPG in Figure~\ref{fig:runningExample}.
The PDA generated by Algorithm~\ref{alg:constr_recog} for this VPG is shown in
Figure~\ref{fig:recog_PDA}.
The input symbols and stack actions are marked on the edges.

\begin{figure}[t]
\[
L\ra\br{A}L \,|\, \epsilon;\ A\ra cC \,|\, cD;\ C\ra cE;\ D\ra dE;\ E \ra \epsilon
\]
\caption{An example VPG.}
\Description{An example VPG.}
\label{fig:runningExample}
\end{figure}
  
\begin{figure}[t]
    \centering
    \includegraphics[width=0.7\textwidth]{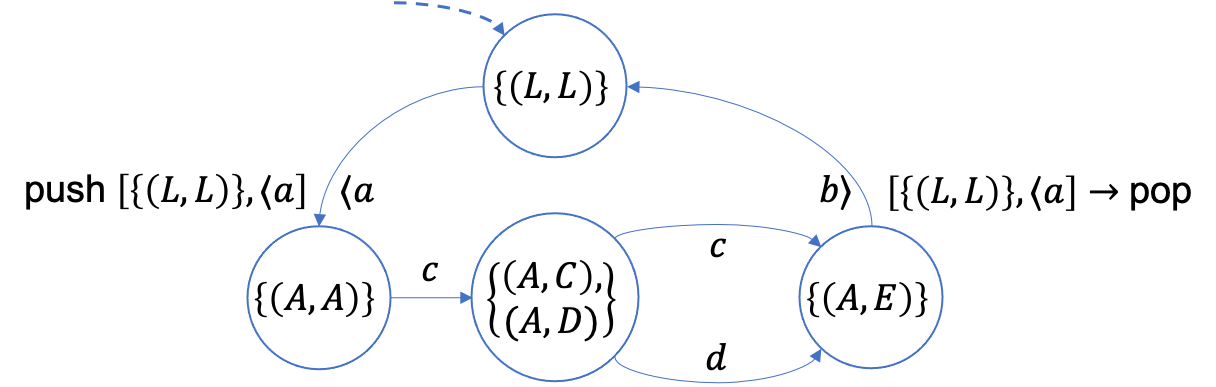}
    \caption{The PDA generated by Algorithm~\ref{alg:constr_recog} for grammar in Figure~\ref{fig:runningExample}. ``push $[\{(L,L),\ac\}]$'' means pushing
        $[\{(L,L),\ac\}]$ to the stack, and ``$[\{(L,L),\ac\}]\ra$ pop'' means removing the top of the stack, 
        when it is $[\{(L,L),\ac\}]$. 
     }
    \Description{An example of the recognizer PDA}
    \label{fig:recog_PDA}
\end{figure}

Once the PDA is constructed from a VPG, it can be run on an
input string in a standard way. Then PDA correctness 
can be stated as follows. We detail PDA execution and the
correctness proof in Appendix~\ref{app:correct_recognizer}.
\begin{theorem}[PDA correctness]
    \label{corecog}
    For VPG $G$ and its start symbol $L_0$, a string $w\in\Sigma^*$ can be derived from $L_0$ (i.e. $L_0\deriv w$) iff $w$ is accepted by the corresponding PDA.
\end{theorem}

For converting general VPGs (i.e., with pending rules) to PDA,
a couple of changes need to be made on the derivative-based approach:
(1) the derivative functions need to consider also pending rules;
(2) the acceptance stack may be nonempty because of pending call symbols.
The construction is discussed in Appendix~\ref{app:recog_general_VPG}.

\section{VPG based parsing} \label{sec:parsing}
Parsing is a process to build the parse trees of a given string based
on an input grammar.  It is equivalent to finding the sequences of
rules that can generate the input string.  
Our VPG-based parsing framework is largely enlightened by the recognizer
construction in Section~\ref{sec:recognition}.
Observe that the execution of the recognizer PDA on an input string $w$ can be 
represented as a trace of runtime configurations:
$(S_0,T_0)\xrightarrow{\textit{$w_1$}}(S_1,T_1)\xrightarrow{\textit{$w_2$}}(S_2,T_2)
\xrightarrow{\textit{$w_3$}}\cdots$,
where $(S_i,T_i)$ is the configuration after consuming $w_i$, the $i^{th}$ symbol of $w$. 
Each transition in the trace is because of a set of possible rules in the input VPG; therefore, we can augment the configuration trace with information about what rules
can be applied during each transition, which gives
$(S_0,T_0)\xrightarrow{\textit{$w_1$}}((S_1,T_1),m_1)\xrightarrow{\textit{$w_2$}}((S_2,T_2),m_2)
\xrightarrow{\textit{$w_3$}}\cdots$,
where $m_i$ is the set of possible rules at step $i$. With the augmented configuration
trace, we can construct parse trees for the input string.

Given this intuition, we could build a parser directly based on the
recognizer.  However, to reduce the burden of formally verifying the
parser, we make several trade-offs in designing the parsing algorithm:
(1) instead of extending the recognizer, we present a way to construct
the parser PDA independently; as will be shown in
Section~\ref{sec:correct_parser}, this allows us to formalize
correctness in a natural way; (2) we replace the context nonterminal
in the recognizer state with a boolean value when constructing the
parser; \secondRound{this simplifies formal verification, but may
  introduce invalid edges in the parse forest; as a result, our design
  adds a pruning step to prune invalid edges.}

\subsection{Overview of VPG-based parsing}\label{subsec:pgPDA}
At a high level, our
VPG-based parsing framework takes an input VPG and generates three
components: (1) a \emph{parser PDA}, which takes an input string and
constructs a parse forest representing possibles parses of the string
according to the VPG; (2) a \emph{pruner PDA}, which takes the parse forest
and removes invalid edges; (3) an \emph{extractor}, which takes the
pruned parse forest and extracts parse trees. 

Before discussing these steps in detail, we use the well-matched
VPG in Figure~\ref{fig:runningExample} to illustrate the steps of
the parser and the pruner PDAs.
Let $w=\br{cd}$ be the string to parse. Figure~\ref{fig:parse_vpg}
visualizes the high-level steps of how our VPG parser parses $w$.
Nodes in the figure are sentential forms with a dot indicating the
parsing position. The prefix before the dot in a sentential form is
the input seen so far; the nonterminal immediately after the dot is
the one to parse next; the remainder of the sentential form is
actually represented by the stack in our parser, but for ease of
understanding we also add it to the nodes in
Figure~\ref{fig:parse_vpg}. The figure also shows for each step the
set of possible rules, those rules represented as parse-tree
edges, and the stack after the step.

\begin{figure}
    \centering
    \includegraphics[width=0.9\textwidth]{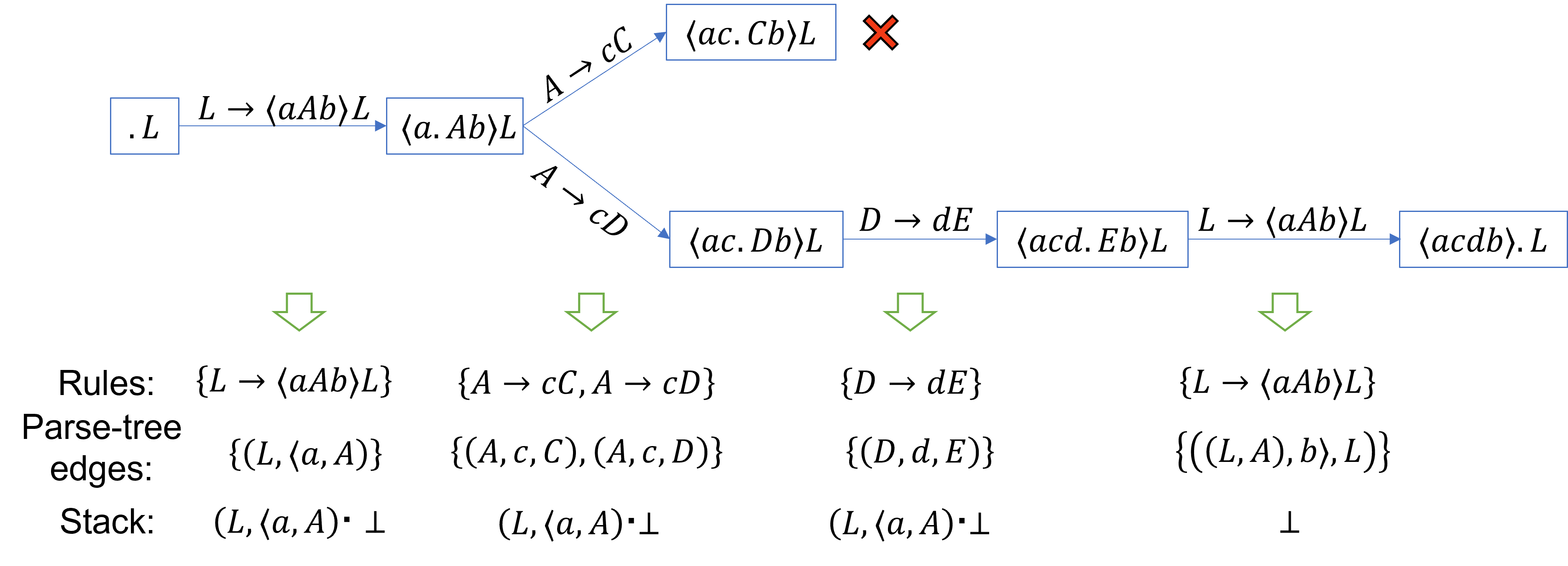}
    \caption{The parsing process for string $\br{cd}$ for the grammar
      in Figure~\ref{fig:runningExample}. There are two traces. The
      first one, followed by a red cross mark, becomes invalid after
      step 2. Possible rules applied in each step are collected in a
      set, which is converted to the set of parse-tree edges at that
      step. The stack keeps track of previous rules that generated
      call symbols, so that later they can be applied to generate the
      corresponding return symbols.}
    \Description{An example of parsing.}
    \label{fig:parse_vpg}
\end{figure}

As an example, after symbol $\ac$, we have the sentential form $\ac
. A \bc L$, where $\ac$ is the already parsed input, $A$ is the
nonterminal to parse, and $\bc L$ is the remainder. For this step,
only one rule is possible: $L\ra \br{A}L$. To represent such a rule
used in building a parse tree, we use a triple instead of the rule
directly; the reason is
to differentiate when $\ac$ is generated by a matching rule
from when $\bc$ is generate by the same
rule. For example, the parsing of $\ac$ using $L\ra \br{A}L$ is
converted to triple $(L, \ac, A)$.  And the remainder is represented by a
stack element $(L,\ac,A)$, which tells us that a rule such as $L\ra
\br{A}L$ was used to create this stack element; with this information,
at step 4 of Figure~\ref{fig:parse_vpg}, 
the parser knows that the next symbol should match $\bc$ and, after that, 
$L$. So the rule $L\ra \br{A}L$ is used twice: at step 1
for $\ac$ and step 4 for $\bc$.

As another example, for the second symbol ``c'' in the input
string, there are two possible rules. One with $C$ as the next
nonterminal to parse, and the other with $D$. However, notice here in
both cases the remainders are $\bc L$ and represented by the same
stack. This example shows the crucial difference between VPG parsing
and CFG parsing. In a general CFG parsing algorithm such as GLR, each
possibility has its own stack, reflecting the nondeterministic nature
of when stacks are changed in CFGs.  In contrast, in VPG parsing all
possibilities in one step share the same stack, enabled by the VPG
property that the stack is changed when consuming only call/return symbols; as a
result, the stack can be shared and factored out.

From this example, we can see that, given an input string $w$ of
length $n$, the parser PDA generates a set of possible rules for
the $i$-th input symbol; since each rule is represented by a triple
in our parser, each step generates a set of triples $m_i$.
We call the sequence of $[m_1, \ldots, m_n]$ the
\emph{parse forest} for $w$; from the sequence, we can recover all
parse trees for $w$. We call $[e_1, \ldots, e_k]$, where $e_i
\in m_i\ \mbox{for } i \in [1..k]$, a \emph{trace} of parse-tree edges;
when $k=n$, it is a \emph{complete trace}. 
A complete trace is a linear representation of a
candidate parse tree. For example, a complete trace in
Figure~\ref{fig:parse_vpg} is $[(L,\ac,A), (A,c,D), (D,d,E), ((L,A),
  \bc, L)]$, representing a candidate parse tree.
Although the number of traces may grow exponentially with the length
of the input string\footnote{For example, consider the grammar
  ``$L\ra \epsilon|cA|cB;\ A\ra dL;\ B\ra dL$'' and the string
  $w=(cd)^{n}$. The number of traces is $O(2^n)$.}, the set of
{distinct} rules possible at each step is always finite and bounded by
the size of the input grammar. As a result, the parse forest
representation is linear to the size of the input string.

Notice that the first trace in Figure~\ref{fig:parse_vpg} is followed
by a red cross mark, because the sentential form $\ac c. C \bc L$
cannot be followed by any rule to generate $d$. However, the parse
forest still keeps the invalid edge $(A,c,C)$. In our approach, we use
a pruner PDA to prune invalid parse-tree edges in the parse forest.  For
example, the pruned parse forest for the above example is
$[\{(L,\ac,A)\},\{(A,c,D)\},\{(D,d,E)\},\{((L,A),\bc,E)\}].$ 

Among the three components of the VPG parsing framework, the first two
can fail: the parser PDA may fail because it cannot
make a transition with the next input symbol; the pruner PDA may
prune the parse forest to an empty parse forest, meaning
that the input string cannot be parsed.

\subsection{The parser PDA} \label{sec_formal_parsing}

As discussed earlier, a parse tree for a VPG can be represented by a
linear sequence of triples, each representing an edge in the parse
tree.  E.g., when rule $L\ra\ac L_1$ is used, it is represented as
$(L,\ac,L_1)$. However, for general VPGs with pending rules, such
triples are insufficient. For example, $(L,\ac,L_1)$ can be the result
from the rule $L\ra\ac L_1$, or the rule $L\ra\br{L_1}L_2$. We need to
further differentiate pending rule edges and matching rule edges,
since pending rules cannot be used within matching rules, required by
general VPGs (Definition~\ref{def:generalVPG}).

Our solution is to tag every nonterminal in a parse-tree edge a
boolean $u$; a similar notion called \emph{linear acceptance} is
discussed by \citet{AlurM09NWA}.  Let $L$ be a nonterminal; a tagged
$L$ is written as $L^u$.  Intuitively, $L^\true$ generates only
well-matched strings and can use only well-matched rules when
generating call/return symbols; in contrast, $L^\false$ can also use
pending rules when generating call/return symbols.  For a general VPG,
if the parser uses a well-matched rule $L\ra\br{L_1}L_2$ to match
$\ac$, then it has to use $L_1^\true$ to perform parsing next, since
general VPGs require that $L_1$ must generate well-matched strings.

With the above discussion, parse-tree edges for general VPGs can be defined as follows:

\begin{definition}[The edges $\mml,\mmc,\mmr$]
    Given a VPG $G=(\Sigma,V,P,L_0)$, 
    \begin{enumerate}
        \item the set of plain edges, denoted as $\mml$, is defined as $\{ (L^u,c,L_1^u) \mid (L\ra c L_1)\in P \}$;
        \item the set of call edges, denoted as $\mmc$, is defined as $\predset{(L^u,\ac,L_1^\true)}{\exists \bc\ L_2,\ (L\ra \br{L_1}L_2)\in P} \cup \predset{(L^\false,\ac,L_1^\false)}{(L\ra \ac  L_1)\in P}$;
        \item the set of return edges, denoted as $\mmr$, is defined as
        $\predset{((L^u,L_1^\true),\bc, L_2^u)}{\exists \ac, (L\ra \br{L_1}L_2)\in P} \cup \predset{(L^\false, \bc, L_1^\false)}{(L\ra \bc  L_1)\in P}$.
    \end{enumerate}
\end{definition}
Note that pending rules $L\ra \ac L_1$ and $L\ra \bc L_1$ can be used
only in edges starting with $L^\false$; further, when using a matching
rule $L\ra \br{L_1}L_2$ to generate edges starting with $L_1$, its tag
must be $\true$ as $L_1$ should match only well-matched strings.  As
we will formalize later, a complete trace that constitutes 
a parse tree must satisfy the following constraints: (1) it must start
with $L_0^\false$, where $L_0$ is the start nonterminal, (2) it must end with some $L_1^\false$ for some $L_1$ such that $(L_1\ra\epsilon) \in P$; the tag
must be false so that no matching rule is waiting to be finished, and 
$L_1\ra\epsilon$ makes sure that no more inputs are
expected to match $L_1$.

With the above definition of parse-tree edges, a parse tree is then a
sequence of plain, call, or return edges. A parse forest is a sequence
$[m_1, \ldots, m_n]$, where each $m_i$ is a subset of $\mml$, $\mmc$, or
$\mmr$.  In the following discussion, we use $\ml$, $\mc$, and $\mr$
for an arbitrary subset of $\mml$, $\mmc$, and $\mmr$, respectively.

\begin{definition}[Parser PDA states and stack]
    Given a VPG, we introduce a parser PDA, where a state, denoted as $m$, is a subset of $\mml$, $\mmc$, or $\mmr$, and each element in the stack $T$ is a subset of $\mmc$.
\end{definition}
It is easy to see that a sequence of parser PDA states constitutes a parse forest.

Similar to the development of the VPG recognizer, we next define three
derivative functions, denoted as $p_c$, $p_{\ac}$, and $p_{\bc}$, to
formalize how the parser PDA makes transitions. From the perspective
of parse trees, each transition extends existing traces for parsing
string $w$ to new traces for parsing string $wi$, assuming $i$ is the
next input symbol.  Notation-wise, we use the placeholder ``$\_$''
to represent an entity whose value does not
matter. For example, edge $(\_,\_,L)\in m$ is defined as $\exists
L_1\ i,\ \st\ (L_1,i,L)\in m$ or $\exists
L_1,L_2,i,L\ \st\ (((L_1,L_2),i,L)\in m$, where only $L$ is of
interest.

\begin{definition}[Derivative functions] 
\label{def:trans_parser}
    Given a VPG $G=(V,\Sigma,P,L_0)$, suppose the current state of the parser PDA is $m$ and the current stack is $T$.
    
    \begin{enumerate}
        \item $p_c(m)=(m',\lambda T.T)$, where
            $m'=\{(L^u,c,L_1^u)\mid (\_,\_,L^u)\in m \land (L\ra c L_1)\in P\}.$

            To generate $c$, we consider each parse-tree edge in $m$. 
            If it is of the form $(\_,\_,L^u)$, find
            every possible rule $L\ra c L_1$ for some $L_1$.  We then
            add edge $(L^u,c,L_1^u)$ to the next state $m'$.
            Intuitively, if the current trace matches $w$ and the
            nonterminal to parse is $L$, then with the extra $c$, the
            new trace matches $wc$ and the new nonterminal to parse is
            $L_1$.  Further, the boolean tag $u$ is passed from $L$ to
            $L_1$, since no matching rule is used at this step.
        
        \item $p_{\ac}(m)=(m',\lambda T.\;m'\cdot T)$, where
          \[
            \begin{array}{l}
                m'= \{(L^u, \ac,L_1^\true) \mid (\_,\_,L^u)\in m\land\exists \bc\ L_2,\ (L\ra \br {L_1}L_2)\in P \}\ \cup \\
                \hspace{5ex} \{(L^\false,\ac,L_1^\false)\mid (\_,\_,L^\false)\in m \land (L\ra \ac  L_1)\in P \}.
            \end{array}
          \]
            If $(\_,\_,L^\true)\in m$, a matching rule
            is waiting to be finished and we cannot use a pending rule
            such as $L\ra \ac L_1$. Thus, only a new matching rule can
            be used to generate $\ac$. So it finds every possible rule
            $L\ra \br{L_1}L_2$ and adds the parse-tree edge
            $(L^\true,\ac,L_1^\true)$ to $m'$. Notice that $L_1$ must
            be tagged with $\true$.  If $(\_,\_,L^\false)\in
            m$, either a matching or a pending rule can be used. The
            matching-rule case is similar to the case when
            $u=\true$. Further, it finds a rule like $L\ra\ac L_1$ for
            some $L_1$ and adds $(L^\false,\ac,L_1^\false)$ to
            $m'$. In addition, the new state $m'$ is pushed to the
            stack to match a possible return symbol at a later point.

        \item $p_{\bc}(m,{\mc})=(m',\tail)$, where ${\mc}=\hd T$ if $T\neq \bot$, and ${\mc}=\emptyset$ if $T=\bot$, and
          \[
            \begin{array}{l}
                m'=
                \{((L^u,L_1^{\true}) ,\bc,  L_2^u)\mid 
                    (L^u, \ac, L_1^{\true} )\in{\mc} \land \exists L_2,(L\ra \br { L_1 } L_2)\in P\}\ \cup \\
                \hspace{5ex} \{(L^\false, \bc,  L_1^\false)\mid 
                    (\_,\_, L^\false)\in m \land (L\ra \bc  L_1)\in P\}.
            \end{array}
          \]

            Consider $(\_,\_,L^u)\in m$. \secondRound{If $u=\false$,
              we must use a pending rule to generate $\bc$. Every rule
              $L\ra \bc L_1$ is converted to edge
              $(L^\false,\bc,L_1^\false)$ and added to $m'$}.
            If $u=\true$, intuitively we can
            only use a matching rule to generate $\bc$. 
            The information of the last
            unfinished matching rule is stored in ${\mc}$, the top of
            the stack. For any $(L^u,\ac,L_1^{\true})\in\mc$, it finds
            a rule $L\ra \br{L_1}L_2$ for some $L_2$, and adds edge
            $((L^u,L_1^{\true}),\bc,L_2^u)$ to the new state $m'$. The
            nonterminal $L_2$ inherits its tag from the tag of $L$.
            Note that in the above formulation the new state is generated based on
            only $\mc$, not the current state $m$; this design can
            generate invalid edges; a later pruner step will remove
            those invalid edges.

        \end{enumerate}
\end{definition}

\paragraph{Constructing the parser PDA}
Similar to constructing the recognizer PDA, the construction of the parser PDA is
the least solution of the following equation. 
\[
\begin{array}{lll}
A & = &
     A \cup~  \predset{m'}{c\in\Sigma_l \land m\in A \land p_c(m)=(m',f)}  \\
  &  & \cup~ \predset{m'}{\ac\in\Sigma_c \land m\in A \land p_{\ac}(m)=(m',f)} \\
  &  & \cup~ \predset{m'}{\bc\in\Sigma_r \land m \in A,\ m''\in (A\cap\Pow(\mmc))\cup\{\emptyset\} \land \ p_{\bc}(m,m'')=(m',f)} \\
\end{array}
\]
Different from the case for VPG recognizers,
$(A\cap\Pow(\mmc))\cup\{\emptyset\}$ is used for $p_{\bc}$, since a
stack element for the parser PDA has to be a state generated for a call symbol and the stack can also be empty.  

With the above equation, we can construct an algorithm for computing
all parser PDA states and state transitions, similar to
Algorithm~\ref{alg:constr_recog}. The differences are that (1) it
starts with a helper state $m_0=\{(L_0^\false,\_,L_0^\false)\}$, where $L_0$ is the
start nonterminal of the input VPG and $\_$ stands for a dummy
nonterminal, and (2) it uses parser derivative functions for deriving
new states and transitions. The result is a parser PDA whose states
are subsets of $\mml$, $\mmc$, or $\mmr$, and transitions between
states are labeled with $(c,f)$, $(\ac,f)$, or $(\bc, m_c, f)$, where
$c$, $\ac$, or $\bc$ is the next input symbol, $m_c$ is the top
stack element, and $f$ is the stack action.
Figure~\ref{fig:parserPDA} shows the parser PDA generated by the algorithm
for the grammar in Figure~\ref{fig:runningExample}.

\begin{figure}
    \centering
    \includegraphics[width=0.6\textwidth]{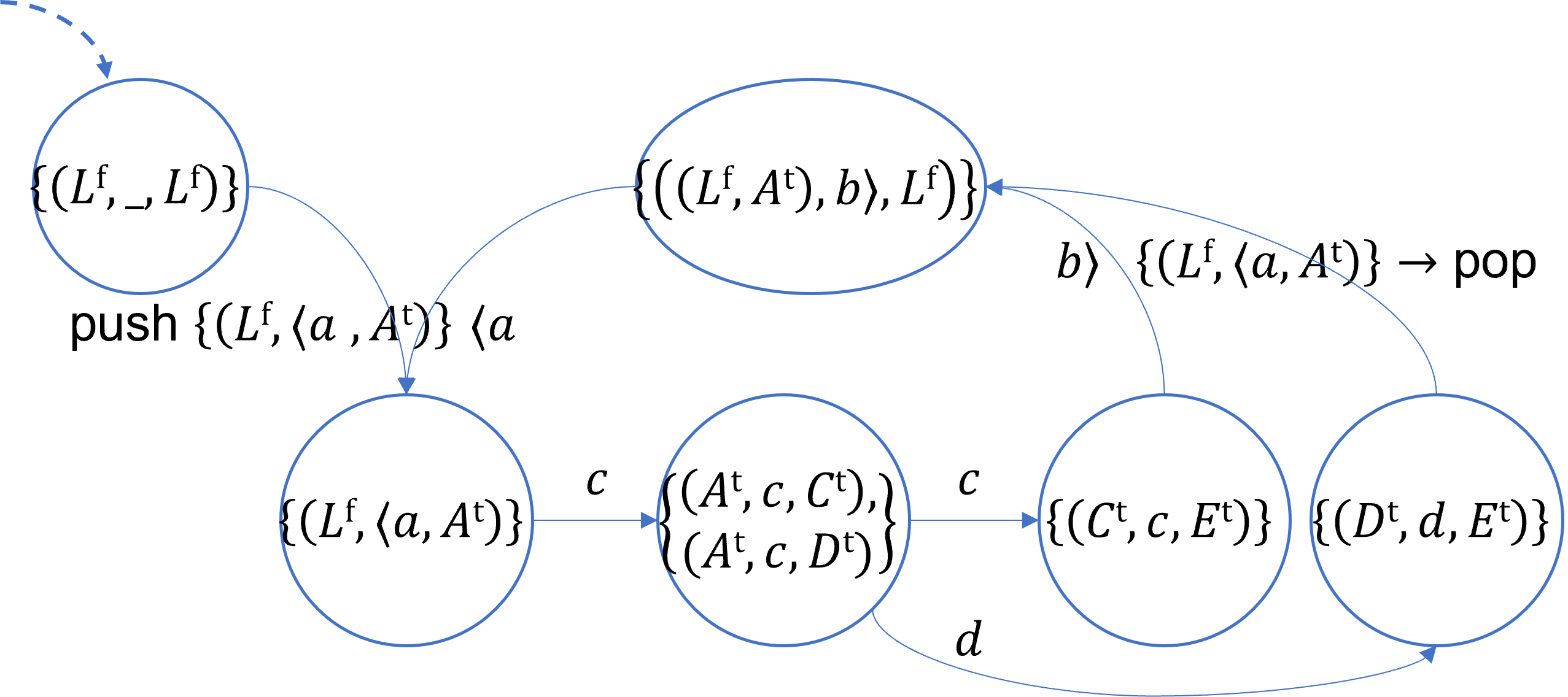}
    \caption{The parser PDA for the grammar in Figure~\ref{fig:runningExample}. We use f for $\false$ and t for $\true$. Finally, the symbols and stack actions of $\ac$ and $\bc$ are each shared by two transitions.}

    \Description{The parser PDA for the grammar in Figure~\ref{fig:runningExample}.}
    \label{fig:parserPDA}
\end{figure}

Given an input string, the parser PDA starts from $(m_0,\bot)$ and
transitions to the next runtime configuration based on the following definition.
\begin{definition}[Runtime transition for parser PDA]
    \label{def:run_pPDA}
    Suppose the current configuration is $(m,T)$ and the next input
    symbol is $i$. The \emph{runtime transition function} $\pPDAtrans(i,m,T)=(m',T')$ of the parser PDA is
    defined as follows.
    \begin{enumerate}
        \item
              if $i\in\Sigma_c\cup\Sigma_l$ and PDA edge $(m,m')$ is
              marked with $(i,f)$, then $\pPDAtrans(i,m,T)= (m',f(T))$;
        \item
              if $i\in\Sigma_r$, $T=m''\cdot T'$, and PDA edge $(m,m')$ is marked with $(i,m'',f)$, then $\pPDAtrans(i,m,T)=(m',f(T))$.
        \item
              if $i\in\Sigma_r$, $T=\bot$, and PDA edge $(m,m')$ is marked with $(i,\emptyset,f)$, then $\pPDAtrans(i,m,T)=(m',\bot)$.
    \end{enumerate}
\end{definition}

\subsection{The pruner PDA and the extractor} \label{sec:trans_g}
\secondRound{The parser PDA parses an input string $w$ of length $n$ and produces a
state trace: $[m_1, \ldots, m_n]$, which can be viewed as a parse
forest. As discussed earlier, it can contain invalid edges. For
example, inside the last state $m_n$, a valid edge must be of the form
$(\_,\_,L_1^\false)$ for some $L_1$ so that $(L_1 \ra \epsilon) \in
P$, signaling the end of parsing; other edges are invalid. Similarly,
if the next symbol to match is a return symbol produced in a
matching rule $L\ra \br{L_1}L_2$, the parse-tree edge immediately before
the one that produces the return symbol must
end with $L_1^\false$ so that $(L_1 \ra \epsilon) \in P$.

After removing some invalid edges, earlier edges in the parse forest
may become invalid. For instance, if pruning removes from $m_n$ an
edge $(L_2^\false,\_,\_)$ and there are no other edges that start with
$L_2^\false$ in the rest of $m_n$, then any $m_{n-1}$ edge that ends
with $L_2^\false$ can also be pruned, since it is not possible to
connect it with an edge in the pruned $m_n$. This is a backward
process. Therefore, our pruner PDA takes the reverse of the parse
forest, $[m_n, m_{n-1}, \ldots, m_1]$, as the input and produces a
pruned parse forest $[m'_n, m'_{n-1}, \ldots, m'_1]$ in reverse.
Further, instead of intervening parsing and pruning steps, we choose
to perform pruning after the parser PDA has finished so that every
state in the parse forest gets pruned only once.

\begin{definition}[Pruner PDA states and stack] 
Given a VPG, we introduce a pruner PDA, where each state, denoted as
$m$, is a subset of $\mml$, $\mmc$ or $\mmr$, and each element of the
stack $T$ is a subset of $\mmr$.
\end{definition}
A stack is needed in the pruner to find valid call edges with respect
to valid return edges in the stack.  The technical details of the
pruner PDA construction are introduced in
Appendix~\ref{app:prunerExtractor}.  After pruning, we get a pruned
parse forest $[m_1', \ldots, m_n']$ and an \emph{extractor} is then
used to extract parse trees from the forest.  If the input VPG is
unambiguous, at most one parse tree can be extracted. The 
definition of $\extractor{[m_1', \ldots, m_n']}$ detailed in
Appendix~\ref{app:prunerExtractor} extracts a \emph{parse-tree set}
$V$, which is a set of parse trees together
with corresponding stacks of call edges.}

\subsection{The correctness proof of the parsing algorithm}
\label{sec:correct_parser}
In this section, we discuss the correctness proof of our core VPG
parsing algorithm; the proof is formalized in the proof assistant Coq.
The correctness theorem is stated based on the relation
$\ptBig{L^u}{w}{v}$, meaning that input string $w$ can be parsed from
nonterminal $L$ with tag $u$ and one of the parse trees is $v$. We
call this relation the \emph{big-step} parse-tree derivation relation,
which is presented in Figure~\ref{fig:big}. Its rules are mostly
straightforward and we explain only the one for $L\ra \br{L_1} L_2$:
it first builds a parse tree for substring $w_1$ with $L_1^\true$ 
since $w_1$ must be a well-matched string; it then builds a
parse tree for substring $w_2$ with $L_2^u$; then a parse tree for
string $\ac w_1 \bc w_2$ can be built by concatenating the parse-tree
edge for $\ac$, the parse tree for $w_1$, the parse-tree edge for
$\bc$, and the parse tree for $w_2$.

\begin{figure}
    \begin{gather*}
        \inference[]
        { (L\ra \epsilon) \in P }
        {\ptBig{L^u}{\epsilon}{\emptyL}} \qquad
        \inference[]
        {(L\ra cL_1) \in P  & \ptBig{L_1^u}{w_1}{v_1}}
        {\ptBig{L^u}{cw_1}{\consL{\left(L^u,c,L_1^u\right)}{v_1}}} \\
        \inference[]
        {(L\ra \ac L_1) \in P & \ptBig{L_1^\false}{w_1}{v_1}}
        {\ptBig{L^\false}{\ac w_1}{\consL{\left(L^\false,\ac,L_1^\false\right)}{v_1}}} 
        \qquad
        \inference[] 
        {(L\ra \bc L_1) \in P & \ptBig{L_1^\false}{w_1}{v_1}}
        {\ptBig{L^\false}{\bc w_1}{\consL{\left(L^\false,\bc,L_1^\false\right)}{v_1}}} \\
        \inference[]
        {(L\ra \br{L_1} L_2) \in P
            & \ptBig{L_1^\true}{w_1}{v_1}
            & \ptBig{L_2^u}{w_2}{v_2}
        }
        {\ptBig{L^u}{\br{w_1}w_2}{\left[\left(L^u,\ac,L_1^\true\right)\right]+v_1+\left[\left(\left(L^u,L_1^\true\right),\bc,L_2^u\right)\right]+v_2}}
        .
    \end{gather*}
    \caption{The big-step parse-tree derivation, assuming a VPG $G=(\Sigma,V,P,L_0)$.}
    \Description{The big-step parse-tree derivation.}
    \label{fig:big}
\end{figure}

\begin{definition}
Suppose $v$ is a trace of parse-tree edges. We define $\firstNT{v}$ to be the starting nonterminal in the
trace and $\lastNT{v}$ to be the last nonterminal in the trace.
That is, 
\[
\begin{array}{l}
\firstNT{v} = \predset{L^u}{v=\consL{(L^u,\_,\_)}{\_}} \\
\lastNT{v} = \predset{L^u}{v= \_ + [(\_,\_,L^u)]}
\end{array}
\]
\end{definition}

\begin{theorem}[Correctness of VPG parsing]
  For an input string $w$ of length $n$, if the parser PDA for VPG
  $G=(\Sigma,V,P,L_0)$ starts with the initial configuration
  $m_0=\set{(L_0^\false,\_,L_0^\false)}$ and $T_0=\bot$ and traverses 
  the following configurations to parse $w$:
$(m_0,T_0)\xrightarrow{\textit{$w_1$}}(m_1, T_1) \cdots \xrightarrow{\textit{$w_n$}}(m_n,T_n)$, and $\extractor{[m_1,\ldots,m_n]} = V$, then
$$\forall v,\ \ptBig{L_0^\false}{w}{v}  \mbox{ iff } \left(\exists E,(v,E)\in V \land \exists L_1,\ \lastNT{v}=L_1^\false \land (L_1 \ra \epsilon) \in P\right).$$
\end{theorem}

\begin{figure}
    \begin{gather*}
        \inference[]
        {
           (L\ra cL_1) \in P &
           v=\emptyL \lor \lastNT{v}=L^u
        }
        {\ptSmall{v}{E}{c}{v+\singletonL{(L^u,c,L_1^u)}}{E}}
        \\
        \inference[]
        {(L\ra \ac L_1) \in P
         & v=\emptyL \lor \lastNT{v}=L^\false}
        {\ptSmall{v}{E}{\ac}
          {v+\singletonL{(L^\false, \ac, L_1^\false)}}
          {(L^\false,\ac, L_1^\false)\cdot E}}
        \\[2ex]
        \inference[]
        {(L\ra \br {L_1}L_2) \in P 
         & v=\emptyL \lor \lastNT{v} = L^u}
        {\ptSmall{v}{E}{\ac}
          {v+\singletonL{(L^u,\ac,L_1^\true)}}
          {(L^u,\ac,L_1^\true)\cdot E}
        }
        \\
        \inference[]
        { (L\ra \bc L_1) \in P
          & v=\emptyL \lor \lastNT{v} = L^\false }
        {\ptSmall{v}{\bot}{\bc}
          {v+\singletonL{(L^\false, \bc ,L_1^\false)}}
          {\bot}
        }
        \\[2ex]
        \inference[]
        { (L\ra \br{L_1}L_2) \in P &
          \lastNT{v} = L_3^\true & (L_3 \ra \epsilon) \in P
        }
        {\ptSmall{v}{(L^u,\ac,L_1^\true)\cdot E}{\bc}
          {v+\singletonL{((L^u,L_1^\true),\bc,L_2^u)}}
          {E}
        }
        \\[2ex]
        \inference[]
        { (L_3\ra \ac  L_4) \in P &
          (L_1\ra \bc L_2) \in P &
          v=\emptyL \lor \lastNT{v}=L_1^\false
        }
        {\ptSmall{v}{(L_3^\false, \ac, L_4^\false) \cdot E}{\bc}
          {v+\singletonL{(L_1^\false,\bc,L_2^\false)}}
          {E}
        }
    \end{gather*}
\caption{The small-step parse-tree derivation, given a VPG $G=(\Sigma,V,P,L_0)$.}
\Description{The small-step parse-tree derivation.}
\label{fig:small}
\end{figure}

\begin{figure}
    \begin{gather*}
        \inference[]
        {}
        {\ptSmallStar{v}{E}{\epsilon}{v}{E}}
        \qquad
        \inference[]
        {\ptSmallStar{v_1}{E_1}{w}{v_2}{E_2}
         & \ptSmall{v_2}{E_2}{i}{v_3}{E_3} }
        {\ptSmallStar{v_1}{E_1}{wi}{v_3}{E_3}}
    \end{gather*}
\caption{The transitive closure of the small-step.}
\Description{The transitive closure of the small-step.}
\label{fig:smallstar}
\end{figure}
To prove the theorem, we need the help of a \emph{small-step}
parse-tree derivation relation so that we can formalize a set of
invariants that are satisfied during each step when running the parser
PDA. The relation $\ptSmall{v}{E}{i}{v'}{E'}$, defined in
Figure~\ref{fig:small}, means that starting with a parse tree $v$ and
a stack of call edges $E$, the parsing of symbol $i$ results in a
new parse tree $v'$ and a new stack of call edges $E'$. In all rules,
$v'$ is the result of adding one new parse-tree edge to $v$, and
therefore it formalizes the process of generating one parse-tree edge
at a time, matching what the parser PDA does. The transitive closure
of the small-step relation is in Figure~\ref{fig:smallstar}.
The following two theorems show the equivalence of big-step and
small-step parse-tree relations.

\begin{theorem}[From big step to small step]
    \label{thm.soundness}
\hspace{10ex}
\begin{enumerate}
\item If $\ptBig{L^\true}{w}{v}$, then $\ptSmallStar{\emptyL}{\bot}{w}{v}{\bot}$.
\item If $\ptBig{L^\false}{w}{v}$, then $\exists E$, \st $\ptSmallStar{\emptyL}{\bot}{w}{v}{E}$.
\end{enumerate}
\end{theorem}

\begin{theorem}[From small step to big step]
  \label{thm.completeness}
  If $\ptSmallStar{\emptyL}{\bot}{w}{v}{E}$, and $\firstNT{v}=L^\false$,  
  $\lastNT{v}=L_1^\false$ and $(L_1\ra\epsilon) \in P$, then $\ptBig{L^\false}{w}{v}$.
\end{theorem}

With the small-step relation, we can formalize those invariants satisfied
by every step when running the parser PDA.  Suppose the parser PDA has consumed
string $w'$ to reach configuration $(m',T')$ and consumes symbol
$i$ next to reach $(m,T)$. Further, the extractor
(Definition~\ref{def:extractor} in Appendix~\ref{app:prunerExtractor})
can extract a parse-tree set $V'$ from the parser PDA state traces for
$w'$.  Finally, extending $V'$ with
one-more parse-tree edge in $m$ to get $V$. Then the following theorem can
be proved: if the previous step
satisfies $\parseInv{L,w',m',T',V'}$, the latest step should also
satisfy $\parseInv{L,w,m,T,V}$, where $w=w'i$. The invariants are
defined as follows.
\begin{prop}[Invariants of VPG parsing]
Let $L$ be a nonterminal, $w$ an input string, $(m,T)$ a parser PDA
configuration, and $V$ a parse-tree set.  The property
$\parseInv{L,w,m,T,V}$ is defined as
    \begin{enumerate}
        \item $\forall {v\ E,}\ (v,E)\in V \mbox{ if and only if }
            \ptSmallStar{\emptyL}{\bot}{w}{v}{E} \land \firstNT{v}=L$;
        \item $\forall (v,E)\in V,\  (E=\bot\ra {T}=\bot)\land (\exists e\ E', E=e \cdot E' \ra e \in \hd {T})$;
         \item $\forall v , \exists E, (v,E)\in V \ra \exists e, e\in m \land v = \_ + [e]$.
    \end{enumerate}
\end{prop}

\secondRound{The above correctness proof is formalized in Coq and
  includes around 3k lines of proofs for the correspondence between
  the big-step and the small-step parse-tree derivations and another
  4k lines for implementing the parser and the parse-tree extractor,
  formalizing the invariants, and proofs for showing that the invariants
  are preserved during parsing.}

\subsection{Time and space complexity}\label{app:complexity}
\secondRound{
When given an input of length $n$, the VPG parser runs two PDAs to
construct the parse forest: a forward parser PDA and a backward pruner
PDA. We
assume their transition tables can be implemented via a data structure
that provides constant-time lookups (e.g., via a hash
table). Therefore, each transition can finish in constant time,
leading to the linear-time running of VPG parsing.

The space complexity depends on the space for representing the
transition tables of the two PDAs. Recall that the transition function
of the parser PDA is $\pPDAtrans(i,m,T)= (m',T')$, where $i\in\Sigma$,
$m$ and $m'$ are states, and $T$ and $T'$ are stacks. The transition
function of the pruner PDA is $\gPDAtrans(m_1,m_2',T) = (m_1',T')$,
where $m_1$, $m_2'$, and $m_1'$ are states, and $T$ and $T'$ are
stacks of states. In fact, only the top of the stack is used by $\mathcal{P}$
and $\mathcal{G}$.  Note since a state is a set of edges and an edge corresponds
to a rule, the size of a state is at most $O(|P|)$.  Thus, there are
at most $O\left(2^{|P|}\right)$ states.  So the total number of transitions in
the two PDAs is bounded by $O\left(2^{3|P|} \times |\Sigma|+2^{4|P|}\right)$,
where $\Sigma$ is the input alphabet. 
{Each entry in the transition table occupies 
$O\left(\log|\Sigma|+|P|\right)$ bits. As a conclusion, 
the space complexity is 
$O\left((\log|\Sigma|+|P|)\times \left(2^{3|P|} \times |\Sigma|+2^{4|P|}\right)\right)$, which is exponential in $|P|$. However, this is the worst-case scenario as not 
all states can be derived; further, it is independent of the input string size.}

In our evaluation, the largest space occupied by the transition
tables is around 1.6 MB, for an HTML grammar (discussed in
Section~\ref{sec:parse_HTML}).}

\section{Designing a surface grammar} \label{sec:surface_grammar}
The format of rules allowed in VPGs is designed for easy studying of
its meta-theory, but is inconvenient for expressing practical
grammars. First, no user-defined semantic actions are allowed. Second,
each VPG rule allows at most four terminals/nonterminals on the right-hand side.  In
this section, we present a surface grammar that is more user-friendly
for writing grammars.  We first discuss embedding semantic
actions. Then we introduce \emph{tagged
  CFGs}, which are CFGs paired with information about how
to separate terminals to plain, call, and return symbols.
We then describe a translator from tagged CFGs
to VPGs. During the conversion, the translator also generates semantic
actions that convert the parse trees of VPGs back to the ones of tagged CFGs.

\subsection{Embedding semantic actions}\label{SecEmbed}
Semantic actions transform parsing results to user-preferred formats.
In a rule $L\ra s_1\cdots s_k$, where $s_k\in\Sigma\cup V$, we treat $L$ as a default action that
takes $k$ arguments, which are semantic values returned by $s_1$ to $s_k$, and returns
a tree with a root node and 
$s_1$ to $s_k$ as children. The prefix notation of a parse tree gives
$$[L,v_{s_1},\cdots, v_{s_k}],$$
where $v_{s_i}$ is the semantic value for $s_i$. The above notation can 
be naturally viewed as a stack machine, where $L$ is an action and $v_{s_i}$ are the values
that get pushed to the stack before the action.
The VPG parse tree can be converted to the prefix notation in a straightforward way.
If we then replace each nonterminal in the tree with its semantic action, the parse tree becomes
a stack machine.

The default action for a nonterminal can be replaced by a user-defined action 
appended to each rule in the grammar.
For example, consider the grammar
$L = c L\mid \br{L} L\mid \epsilon$. 
Suppose we want to count the number of the symbol $c$ in an input string;
we can specify semantic actions in the grammar as follows.
$$
L\ra cL\ @\{\text{let }f_1\ v_1\ v_2 = 1+v_2\} 
   \mid \br{L} L \ @\{\text{let }f_2\ v_1\ v_2\ v_3\ v_4 = v_2+v_4\} 
   \mid \epsilon \ @\{\text{let }f_3\ () = 0\}.
$$

In the above example, a semantic action is specified after each rule, e.g., ``@\{let $f_1\ v_1\ v_2 = 1+v_2$\}''. In the actions, $v_1$, $v_2$, $v_3$ and $v_4$ represent the semantic values returned by the right hand side symbols of the rule. For example, the first semantic action $f_1\ v_1\ v_2\ =\ 1+v_2$ accepts two semantic values $v_1$ and $v_2$, where $v_1$ is returned by $c$ and $v_2$ is returned by $L$. 

As an application, the next subsection shows how to use semantic actions to convert the parse trees of a VPG to the parse trees of its original tagged CFG.

\subsection{Translating from tagged CFGs to VPGs}\label{sec:translator}
Grammar writers are already familiar with CFGs, the basis of
many parsing libraries. We define \emph{tagged CFGs} to be CFGs
paired with information about how to partition terminals into
plain, call, and return symbols
($\Sigma=\Sigma_l\cup\Sigma_c\cup\Sigma_r$);\footnote{We note our
  implementation of tagged CFGs additionally supports regular operators
  in the rules; these regular operators can be easily desugared and
  we omit their discussion.} that is, in a tagged CFG, a terminal
is tagged with information about what kind of symbols it is. Compared
to a regular CFG, the only additional information in a tagged CFG is
the tagging information; therefore, tagged CFGs provide a convenient
mechanism for reusing existing CFGs and developing new grammars in a
mechanism that grammar writers are familiar with. Appendix
\ref{app:Examples} shows some example tagged CFGs.

However, not all tagged CFGs can be converted to VPGs. We use a
conservative validator to determine if a tagged CFG can be converted
to a VPG and, if the validator passes, translate the tagged CFG to a
VPG.  For simplicity, we assume every call symbol is matched with a
return symbol in the input tagged CFG.

The translation steps are summarized as follows:
$$\textit{A tagged CFG}\ra \textit{Simple form}\xrightarrow{\textit{If valid}} \textit{Linear form}\ra \textit{VPG}.$$
At a high level, a tagged CFG is first translated to a simple form,
upon which validation is performed. If validation passes, the
simple-form CFG is translated to a linear-form CFG, which is finally
translated to a VPG. We next detail these steps.

\begin{definition}[Simple forms]
  A rule is in the simple form if it is of the form $L\ra \epsilon$, or of the form
    $L\ra s_1\cdots s_k$,
    where $s_i\in \Sigma_l\cup V$ or $\exists \ac,\bc,L_i,\st\ s_i=\br{L_i}$, $i=1..k$, $k\geq 1$.  A tagged CFG $G=(V,\Sigma,P,L_0)$ is in the simple form, if every rule in $P$ is in the simple form.
\end{definition}
Compared to a tagged CFG, a simple-form CFG requires that there must be a
nonterminal between a call symbol and its matching return symbol.
The conversion from a tagged CFG to a simple-form CFG is
straightforward: for each rule, we replace every string $\br{s}$,
where $\ac$ is matched with $\bc$ and $s\in(\Sigma\cup V)^*$, with
$\br{L_s}$ and generate a new nonterminal $L_s$ and a new rule $L_s\ra
s$. 
After this conversion, a string in the from of $\br{L}$ can be viewed
as a ``plain symbol''; this is a key intuition
for the following steps. We call $\br{L}$ a \emph{matched token} in
the following discussion.

The validation can then perform on the simple form, using its dependency graph.
\begin{definition}[Dependency graphs]
    The dependency graph  of a grammar $G=(V,\Sigma,P,L_0)$ is $(V,E_G)$, where 
    $E_G=\predset{(L,L')}{\exists s_1,s_2\in (\Sigma\cup V)^*,\st\ (L\ra s_1 L' s_2)\in P}.$
\end{definition}
The validator checks for every loop in the dependency graph, either
(1) in the loop there is an edge $(L,L')$ that is produced
from a rule of the form $L\ra s_1\br{L'}s_2 $, where $s_1,s_2\in
(\Sigma\cup V)^*$; or (2) every edge $(L,L')$ in the loop is produced
from a rule of the form $L\ra sL'$, 
$s\in (\Sigma\cup V)^*$
and at least one edge in the loop satisfies $s\not\ra^*\epsilon$.

Once the validation passes, the translation converts a simple-form CFG
to a linear-form CFG.
\begin{definition}[Linear forms]
    A rule is in the linear form if it is in one of the following forms:
        (1) $L\ra\epsilon$;
        (2) $L\ra t_1\cdots t_k$;
        (3) $L\ra t_1\cdots t_kL'$;
    where $t_i\in\Sigma_l$ or $\exists \ac,\bc,L_i,\st\ t_i=\br{L_i}$, $i=1..k$, $k\geq 1$. A tagged CFG $G=(V,\Sigma,P,L_0)$ is in the linear form if every rule in $P$ is in the linear form.
\end{definition}
Note that in a linear-form rule, $t_i$ cannot be a nonterminal, while
in a simple-form rule $s_i$ can be a nonterminal. Further, the linear
form allows rules of the form $L\ra t_1\cdots t_kL'$, where $t_i$ is a
terminal or a matched token.  The main job of the translator is to
convert simple-form rules to linear-form rules.
Appendix~\ref{app:terminate} shows the translation algorithm.

The translation from a linear-form CFG to a VPG is simple.  E.g., for
a rule of the form $L\ra t_1\cdots t_k$, it is translated to 
$L \ra t_1L_1; L_1 \ra t_2 L_2; \ldots; L_k \ra t_kL_k; L_k \ra \epsilon$,
where $L_1$ to $L_k$ are a set of new nonterminals.

All transformations are local rewriting of rules and as a result
it is easy to show that each transformation step preserves the set of
strings the grammar accepts.
\secondRound{We further note that not all tagged CFGs can be converted to VPGs.  For
  example, grammar ``$L\ra cLc|\epsilon$'' cannot be converted since its
  terminals cannot be suitably tagged: intuitively $c$ has to be both
  a call and a return symbol.  Further, since our validation algorithm
  is conservative, it rejects some tagged CFGs that have VPG
  counterparts.  For example, grammar ``$L\ra Lc|\epsilon$'' is
  rejected by the validator since it is left recursive. However, it
  can be first refactored to ``$L\ra cL|\epsilon$'', which is accepted
  by our validator.}

\paragraph{Generating semantic actions} During the conversion, each time the translator rewrites a rule, a corresponding semantic action is attached to the rule. Initially, every rule is attached with one default semantic action. For example, the rule $L\ra AbCd$ is attached with $L^4$, written as 
$L\ra AbCd \ @L^4.$
As mentioned in Section \ref{SecEmbed}, $L^4$ is the default semantic action for constructing a tree with a root node and children nodes that are constructed from semantic values from the right hand side of the rule. The superscript $4$ is its arity. 
During conversion, every time we rewrite a nonterminal $L$ in a rule
$R$ with the right-hand side of rule $L \ra s$, the semantic values
for $s$ are first combined to produce a semantic value for $L$, which
is then used to produce the semantic value for the left-hand
nonterminal of $R$.  If a helper nonterminal $L_s$ is introduced
during conversion and a rule $L_s \ra s$ is generated, we do not
generate a semantic value for $L_s$ but leave the semantic values for
$s$ on the stack so that any rule that uses $L_s$ can use those
semantic values directly.  In this way, we can convert a parse tree of
a VPG to the parse tree of its corresponding tagged CFG.
Appendix~\ref{app:gen_semanticactions} shows an example of the
translation.

\section{Evaluation} \label{sec:eval}

We implemented our VPG parsing library in OCaml. The implementation
used hash tables to store the transition tables of the generated
parser and pruner PDAs to get constant-time lookup.
\secondRound{We evaluated our implementation for the following questions: (1) how applicable
VPG parsing is in practice? (2) what is the performance of VPG parsing
compared with other parsing approaches?

We performed a preliminary analysis for a set of ANTLR4 grammars
in a grammar repository\footnote{\url{https://github.com/antlr/grammars-v4}}.  
Among all 239 grammars, 136 (56.9\%)
grammars could be converted to VPGs by our tagged-CFG-to-VPG translation, after we
manually marked the call and return symbols for those grammars.  Note
that it does not mean the rest cannot be converted; e.g.,
34 grammars cannot be converted because they have
left recursion and the conversion may become possible if the
left recursion is removed. 
We left a further analysis for future work.}

For performance evaluation, we compared 
our VPG parsers
with ANTLR4\footnote{\url{https://www.antlr.org/}.},
a popular parser generator that implements an efficient
parsing algorithm called ALL(*)~\cite{parr2014adaptive}.  
\secondRound{The ALL(*)
algorithm can perform an unlimited number of lookaheads to resolve
ambiguity and it has a  worst-case complexity of $O(n^4)$; however, it
exhibits linear behavior on many practical grammars.} 
\secondRound{We
  also compared the VPG parsers with a few hand-crafted parsers specialized for
  parsing JSON and XML documents, including four mainstream JavaScript
  engines and four popular XML parsers.  Before presenting the
  performance evaluation, we list some general setups: }
\begin{enumerate}
  \item 
\secondRound{During evaluation, we adapted the grammars for JSON, XML, and HTML from ANTLR4\footnote{
  \url{https://github.com/antlr/grammars-v4/blob/master/json/},
  \url{https://github.com/antlr/grammars-v4/blob/master/xml/},
and
\url{https://github.com/antlr/grammars-v4/tree/master/html}}
to tagged CFGs, generated VPG parsers, and compared VPG parsers with the parsers generated by ANTLR in performance.
Appendix~\ref{app:Examples} shows the tagged CFGs for JSON, XML and HTML.}

  \item 
  \secondRound{When comparing with ANTLR, we compared only the parsing time, omitting
  the lexing time.  This is because we used ANTLR's lexers to generate
  the tokens for both VPG parsers and ANTLR parsers.}

\end{enumerate}

\subsection{Comparison with ANTLR on parsing JSON files}\label{sec:parse_JSON}
The JSON format allows objects to be nested within objects and arrays;
therefore, a JSON object has a hierarchically nesting structure, which
can be naturally captured by a VPG. In particular, since in JSON an
object is enclosed within ``\{'' and ``\}'' and arrays within ``[''
  and ``]'', its VPG grammar treats ``\{'' and ``[`` as call 
symbols and treats ``\}'' and ``]'' as return symbols.

\begin{figure}
    \begin{minipage}{0.46\textwidth}
        \centering
        \includegraphics[width=\textwidth]{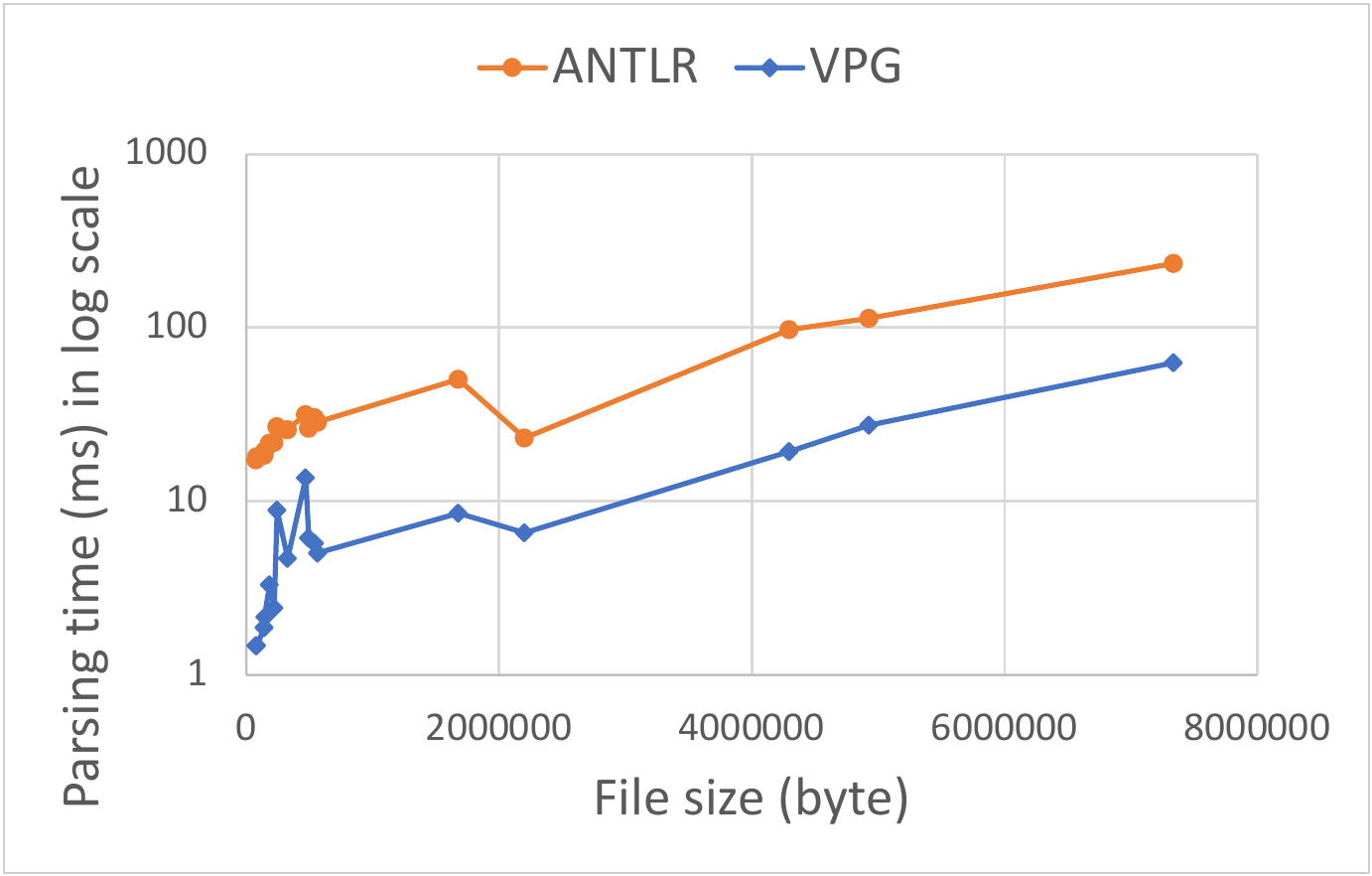}
    \end{minipage}\hfill
    \begin{minipage}{0.54\textwidth}
        \centering
        \begin{tabular}{crrrrr}
            \toprule
            Name	              & Size &   ANTLR     & VPG      & Conv \\
            \midrule
            JSON.parse            & 7.0 MB & 235 ms   & 63 ms & 61 ms \\
            airlines              & 4.7 MB & 113 ms   & 27 ms & 27 ms \\
            educativos            & 4.1 MB &  98 ms   & 19 ms & 19 ms \\
            canada                & 2.1 MB &  23 ms   & 15 ms & 20 ms \\
            citm\_catalog         & 1.6 MB &  50 ms   &  9 ms &  8 ms\\
            \bottomrule
        \end{tabular}
    \end{minipage}
    \caption{Parsing times of JSON files (in log scale on the left).}
    \Description{Parsing times of JSON files.}
    \label{fig:evalJSON}
\end{figure}

When building a VPG parser for JSON, we reused ANTLR's lexer.
Therefore, the evaluation steps are as follows.
\secondRound{$$\textit{Input file}\xrightarrow{\textit{ANTLR Lexer}}\textit{ANTLR tokens}
\xrightarrow{\textit{ANTLR Parser or VPG Parser}} \textit{Results}.$$}

\secondRound{For evaluation, we collected 23 real-world JSON files
from the awesome-json repository,
the nativejson benchmarks, and the \texttt{JSON.parse} benchmarks\footnote{
\url{https://github.com/jdorfman/awesome-json-datasets},
\url{https://github.com/miloyip/nativejson-benchmark}, and
\url{https://github.com/GoogleChromeLabs/json-parse-benchmark}}.
The sizes of the files range from 14 KB to 7 MB.
The parsing times are shown in Figure~\ref{fig:evalJSON};
  note that the y-axis of the left figure (and other figures in this
  section) is in the log scale for better visualization.  As can be
  seen, the VPG parser runs much faster than ANTLR.
  The right of Figure~\ref{fig:evalJSON} shows the VPG parsing times for the 5 largest
  files in our test set; the VPG column is the amount of time cost by running the parser
  and pruner PDAs. On those large files, VPG parsing is about 4 times faster than
  ANTLR. For smaller files, the gap is even larger; Appendix~\ref{app:eval} shows the results for the full test set.}

A downstream application that uses the ANTLR's JSON parser may
wish to keep working on the same parsing result produced by ANTLR's
parser. Therefore, we implemented a converter to convert the parse
forest produced by our VPG parser to ANTLR's parse tree for the input
files. When the grammar is unambiguous, \secondRound{which is the case for the
JSON grammar (as well as the XML and HTML grammars)}, the parse forest is really the encoding of a single
parse tree.  The algorithm of how to convert a VPG parse tree to a stack
machine and how to evaluate the stack machine have been discussed in
Section~\ref{sec:surface_grammar}.  The result of the evaluation 
is a structure that can be directly printed out and compared with; the
same applies to the ANTLR parse tree\footnote{By the ``ANTLR parse tree'',
  we mean the string output by the ANTLR parser with the option
  ``-tree''.}.
The conversion steps are summarized as follows.
\secondRound{$$
\textit{VPG parse tree}\xrightarrow{\textit{Embed actions}}  \textit{Stack machine}\xrightarrow{\textit{Evaluate}} \textit{ANTLR parse tree}.
$$}
Note that in practice this conversion may not
be necessary. A downstream application can directly work on the VPG
parse tree. We include the time to show the conversion time for our
VPG parser to work directly with legacy downstream applications.
\secondRound{The time of conversion is shown in 
the ``Conv'' column on the right hand side of Figure~\ref{fig:evalJSON}.}

\secondRound{\subsection{Comparison with ANTLR on parsing XML files}}
\label{sec:parse_xml}
\secondRound{XML also has a well-matched nesting structure with explicit
start-tags such as \texttt{<p>} and matching end-tags such as
\texttt{</p>}.
However, compared to JSON, there is an additional
complexity for the XML grammar, which makes it necessary to adapt the
XML grammar provided by ANTLR. In particular, the XML lexer in ANTLR
treats an XML tag as separate tokens; e.g., \texttt{<p>} is converted
into three tokens: \texttt{<}, \texttt{p}, and \texttt{>}. Those tokens
then appear in the ANTLR XML
grammar. Part of the reason
for this design is because the XML format allows additional
attributes within a tag; e.g., \texttt{<p id=1>} is a start-tag with an
attribute with name \texttt{id} and value 1.
Below is a snippet of the related XML grammar in ANTLR.}
\begin{verbatim}
element : 
  '<' Name attribute* '>' content '<' '/' Name '>' | '<' Name attribute* '/>' ;
\end{verbatim}
To expose the nesting structure within XML, we add an additional step 
between ANTLR lexing and VPG parsing.
$$\textit{Input file}\xrightarrow{\textit{ANTLR Lexer}}\textit{ANTLR tokens}
\xrightarrow{\textit{VPG Lexer}}
\textit{VPG tokens}
\xrightarrow{\textit{VPG Parser}} Results.$$

The step of VPG lexing coalesces tokens for a single XML tag into a 
single token.  For example,
\texttt{<p>} becomes a single token and is marked as a call symbol. 
For attributes inside tags, they are processed
and attached as tags' semantic values for the following parsing
step. The following shows a snippet of our adapted XML grammar.
\begin{verbatim}
element : <TagOpen content TagClose> | TagSingle ;
\end{verbatim}
The VPG tokens are declared as follows.
\begin{verbatim}
TagOpen = '<' Name attribute* '>' ;
TagClose = '<' '/' Name '>' ;
TagSingle = '<' Name attribute* '/>' ;
\end{verbatim}

\begin{figure}
    \begin{minipage}{0.48\textwidth}
        \centering
        \includegraphics[width=\textwidth]{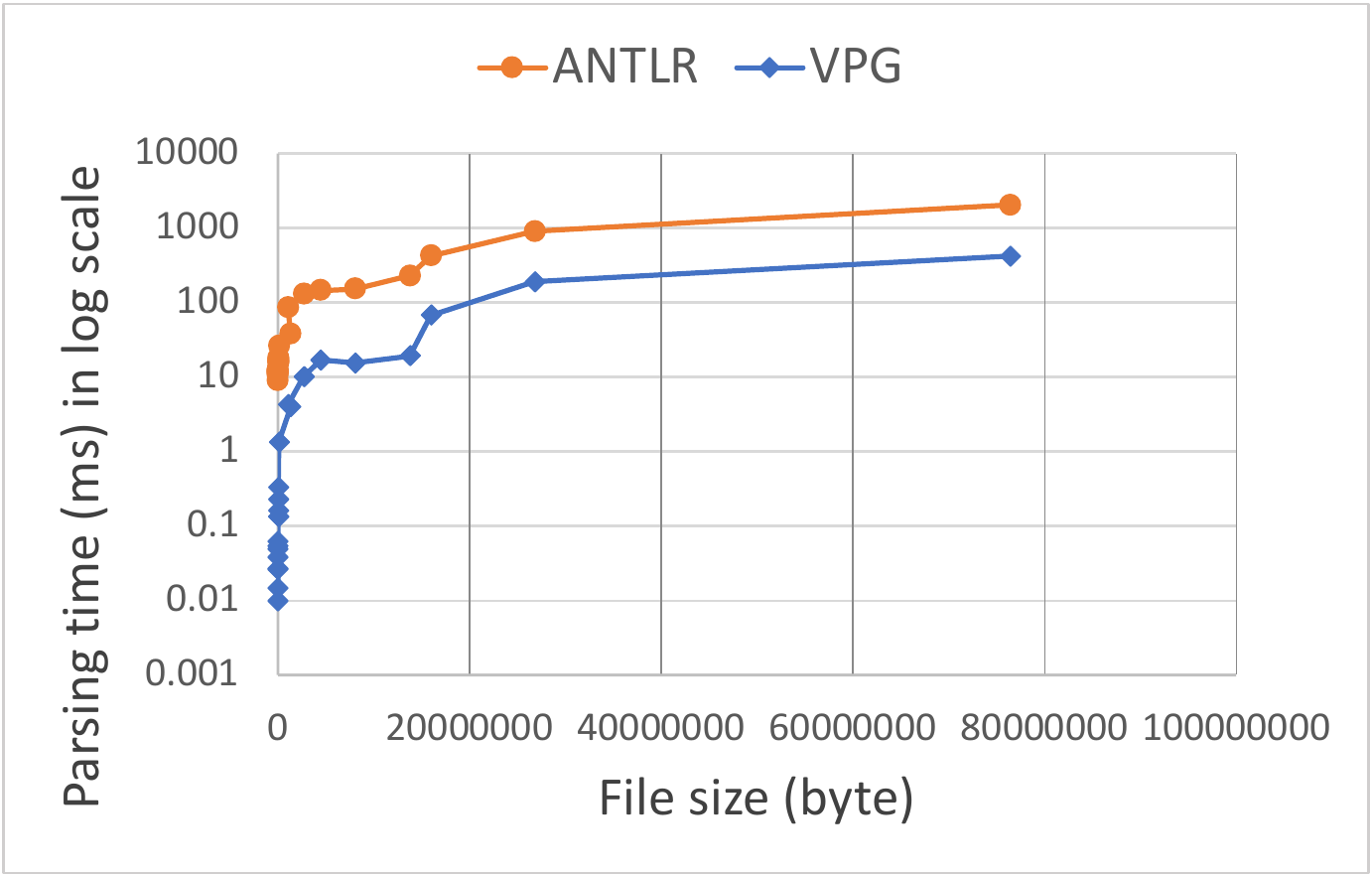}
    \end{minipage}\hfill
    \begin{minipage}{0.52\textwidth}
        \centering
        \begin{tabular}{crrrrr}
            \toprule
            Name	     & Size &	ANTLR    &     VPG  & Conv \\
            \midrule
            po       &  73 MB &  2058 ms  & 425 ms & 1070 ms \\
            cd       &  26 MB &   913 ms  & 192 ms & 455 ms \\
            address  &  15 MB &   429 ms  &  68 ms & 139 ms \\
            SUAS     &  13 MB &   232 ms  &  19 ms & 13 ms \\
            ORTCA    & 7.7 MB &   154 ms  &  16 ms & 6 ms \\
            \bottomrule
        \end{tabular}
    \end{minipage}\hfill
        \caption{Parsing times of XML files (in log scale on the left).}
        \Description{Parsing times of XML files.}
        \label{fig:evalXML}
\end{figure}

\secondRound{For evaluation, we used the real-world XML files provided by the
VTD-XML
benchmarks\footnote{\url{https://vtd-xml.sourceforge.io/2.3/benchmark_2.3_parsing_only.html}.},
which consist of a wide selection of 23 files ranging from 1K to
73MB.  The parsing times are presented in Figure~\ref{fig:evalXML};
Appendix~\ref{app:eval} shows the results for the full test
set.
Similar to JSON, VPG parsing on XML files is much faster than ANTLR parsing;
on the 5 largest XML files, VPG parsing is about 5 times faster; on smaller
files, the gap is even larger.}

\subsection{Comparison with ANTLR on parsing HTML files}\label{sec:parse_HTML}
A snippet of the HTML grammar in ANTLR is listed below:
\begin{verbatim}
htmlElement: 
  ‘<’ TAG_NAME htmlAttribute* (‘>’ (htmlContent ‘<’ ‘/’ TAG_NAME ‘>’)? | ‘/’ ‘>’ ) ;
htmlContent: htmlChardata? ((htmlElement | CDATA | htmlComment) htmlChardata?)* ;
\end{verbatim}
Similar to the XML grammar, the HTML grammar allows self-closing tags
such as \texttt{<br/>}. However, the HTML grammar in addition allows
optional end tags, which is not allowed in XML. For example, 
the HTML tag \texttt{<input type="submit" value="Ok">}
cannot have a matching end tag according to the HTML standard.
Although this kind of tags is also ``self-closing'', we will use the
terminology of optional end tags since that is how the official HTML5
standard describes it.  As will be shown in our experimental data, the
complexity in this grammar makes ANTLR's parsing of HTML files
extremely slow.

Similar to the XML case, we introduced a VPG lexer to coalesce tokens
for a single tag into a single token. However, the optional end-tags
introduce additional complexity. To explain, let us first examine
the relevant part of the VPG grammar:
\begin{verbatim}
htmlElement = TagPlain | <TagOpen htmlElement TagClose> | TagSingle ;
\end{verbatim}
The VPG tokens are declared as follows.
\begin{verbatim}
TagPlain = ‘<’ TAG_NAME htmlAttribute* ‘>’ ;
TagOpen = ‘<’ TAG_NAME htmlAttribute* ‘>’ ;
TagClose = ‘<’ ‘/’ TAG_NAME ‘>’ ;
TagSingle = ‘<’ TAG_NAME htmlAttribute* ‘/’ ‘>’ ;
\end{verbatim}
\texttt{TagPlain} is for HTML tags that cannot have matching end-tags, 
and \texttt{TagSingle} is for self-closing tags.  The VPG lexer
first merges ANTLR tokens related to a single tag, and then determines
which tags are call symbols, return symbols, and plain symbols. A
start tag with no matching end tag is marked as  a plain symbol in this
process. This is implemented with a straightforward method: the first
$k$ open HTML tags are matched with the last $k$ close HTML tags, and
the rest open HTML tags are viewed as plain symbols, where $k$ is the
number of close HTML tags in the file (the number of close tags is
always less than or equal to the number of open tags).

\begin{figure}[t]
    \begin{minipage}{0.5\textwidth}
        \centering
        \includegraphics[width=\textwidth]{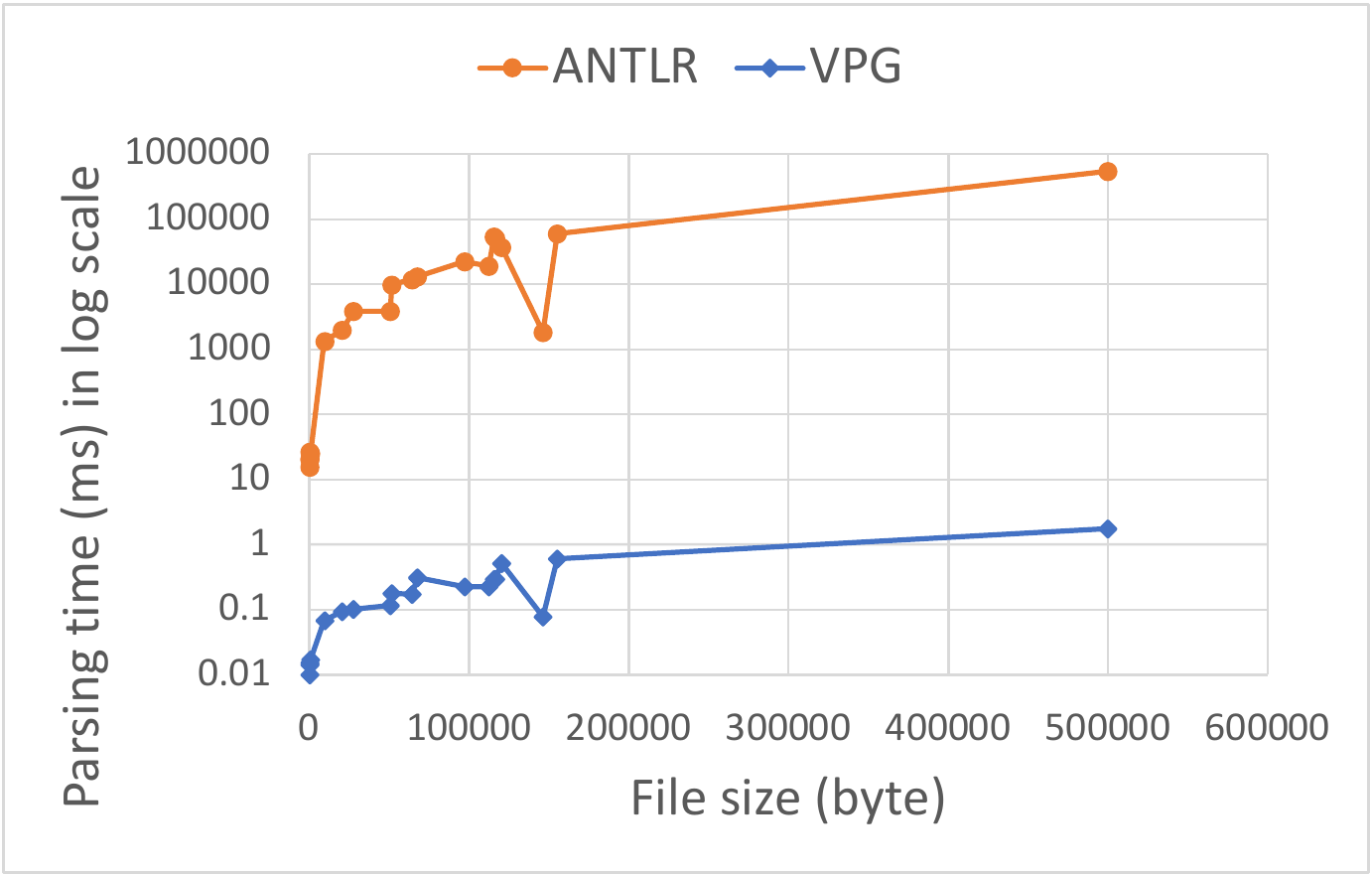}
    \end{minipage}\hfill
    \begin{minipage}{0.5\textwidth}
        \centering
        \begin{tabular}{crrrrr}
            \toprule
            Name     & Size & ANTLR   & VPG & Conv   \\
            \midrule
            youtube   & 489 KB &  543 s   & 1.8 ms &  2.3 ms \\
            digg      & 152 KB &   60 s   & 0.6 ms &  0.7 ms \\
            cnn1      & 118 KB &   37 s   & 0.5 ms &  0.5 ms \\
            reddit2   & 114 KB &   52 s   & 0.3 ms &  0.6 ms \\
            reddit    & 114 KB &   54 s   & 0.3 ms &  0.6 ms \\
            \bottomrule
        \end{tabular}
    \end{minipage} \hfill
        \caption{Parsing times of HTML files (in log scale on the left).}
        \Description{Parsing times of HTML files.}
        \label{fig:evalHTML}
\end{figure}

For evaluation, we used the 19 real-world HTML files provided in ANTLR's
repository\footnote{\url{https://github.com/antlr/grammars-v4/tree/master/html/examples}}.
The parsing times are presented in Figure~\ref{fig:evalHTML}.
The conversion times of the parse trees are shown in the “Conv” column.
As we can see, our VPG parser significantly outperforms the ANTLR
parser, with more than 4 orders of magnitude of difference.  
We emphasize that in our evaluation the VPG
parser and the ANTLR parser accept the same HTML files, and they
produce the same parse trees with the help of a converter.

From the ANTLR profiling tool,  we found that around 99\% time cost by ANTLR is in the prediction of
\begin{verbatim}
    (htmlContent TAG_OPEN TAG_SLASH TAG_NAME TAG_CLOSE)?
\end{verbatim}
in the rule of ``htmlElement'', which triggers many lookahead symbols, and also many DFA cache misses (around a miss rate of 90\%). Appendix~\ref{app:profile} shows the profiler result for ``bbc.com.html''.
This HTML evaluation shows the power of VPGs in
designing practical language parsers, due to their ability of linear-time parsing.

\paragraph{Summary of comparison with ANTLR}
Our performance evaluation shows that our VPG parsing library
generates parsers that run significantly faster than those generated
by ANTLR on grammars that can be converted to VPGs, such as JSON, XML,
and HTML.

\subsection{Comparison with hand-crafted parsers}

\begin{table}
\centering
\caption{Parsing times of 5 largest JSON files. ``SpiderM'' stands for ``SpiderMonkey'', and "JSCore"  for ``JavaScriptCore''.}
\label{tab:evalJSON2}
\begin{tabular}{crrrrrrrr}
    \toprule
    Name        & Size& \begin{tabular}{@{}c@{}}ANTLR \\ Lex\end{tabular} & \begin{tabular}{@{}c@{}}VPG \\ Parse\end{tabular}& Lex+Parse& SpiderM & JSCore& V8  & Chakra \\
    \midrule
JSON.parse  & 7.0 MB & 130    ms & 63     ms & 193    ms & 118      ms & 139        ms & 76 ms& 88 ms \\
   airlines & 4.7 MB & 81 ms & 27 ms & 108 ms & 74 ms & 95 ms & 42 ms & 56 ms \\
educativos  & 4.1 MB & 108    ms & 19     ms & 128    ms & 71       ms & 421        ms & 45 ms& 49 ms \\
canada      & 2.1 MB & 45     ms & 15     ms & 60     ms & 57       ms & 68         ms & 34 ms& 44 ms \\
citm\_catalog & 1.6 MB & 39     ms & 9      ms & 47     ms & 34       ms & 71         ms & 28 ms& 25 ms \\
    \bottomrule
\end{tabular}
\end{table}
\begin{table}
\caption{Parsing times of 5 largest XML files. ``HP2'' stands for ``HTMLParser2''.}
\label{tab:evalXML2}
\begin{tabular}{crrrrrrrr}
    \toprule
    Name        & Size& \begin{tabular}{@{}c@{}}ANTLR \\ Lex\end{tabular}       &  \begin{tabular}{@{}c@{}}VPG \\ Parse\end{tabular}   & Lex+Parse     & Fast-XML& Libxmljs& SAX-JS  & HP2 \\
    \midrule
    po     & 73 MB & 1812 ms& 425 ms& 2238 ms& 3278 ms& 897 ms& 6618 ms& 1827 ms  \\
    cd     & 26 MB & 732 ms & 192 ms& 924 ms & 1298 ms& 419 ms& 2103 ms& 735 ms  \\
    address& 15 MB & 367 ms & 68 ms & 435 ms & 584 ms & 196 ms& 1012 ms& 331 ms \\
    SUAS   & 13 MB & 237 ms & 19 ms & 256 ms & 254 ms & 182 ms& 1214 ms& 169 ms \\
    ORTCA  & 7.7 MB& 142 ms & 16 ms & 157 ms & 138 ms & 91 ms & 665 ms & 89 ms \\
    \bottomrule
\end{tabular}
\end{table}

We also compared VPG parsers with hand-crafted parsers for
JSON and XML documents.
For JSON, we compared with four mainstream JavaScript engines 
(V8, Chakra, JavaScriptCore, and SpiderMonkey) and evaluated them on
the JSON files in Section~\ref{sec:parse_JSON}.
For XML, we compared with four popular XML parsers (fast-xml-parser, libxmljs, sax-js, and htmlparser2)
\footnote{\url{https://github.com/NaturalIntelligence/fast-xml-parser\#readme}, 
\url{https://github.com/libxmljs/libxmljs},
\url{https://github.com/isaacs/sax-js}, and
\url{https://github.com/fb55/htmlparser2}.}, and evaluated them with
the XML files in Section~\ref{sec:parse_xml}.

The evaluation results for the largest files are shown in Table
\ref{tab:evalJSON2} and Table \ref{tab:evalXML2}; the full results are
in Appendix~\ref{app:eval}. Note that the hand-crafted parsers can
process raw texts directly, while our VPG parsers process the tokens
generated by ANTLR's lexers.  Therefore, we show separately the lexing
time of ANTLR (column ``ANTLR Lex''), the parsing time of VPG parsing
(column ``VPG Parse''), and the combined time (column ``Lex+Parse'').
From the results, we can see that although the total time of VPG
parsing is not the shortest among all parsers, the parsing time alone
is.  Thus, VPG parsers show promising potential in performance, with
additionally verified correctness over hand-crafted parsers.
The total parsing time can be reduced by replacing ANTLR's lexer
with a faster, customized lexer, since the parsing time of VPG is
shorter than the lexing time.  Also, combining the lexing and parsing
steps, as is common in hand-crafted parsers, can usually improve the
overall time.

\section{Limitations and Future Work}
\secondRound{As noted earlier, the correctness of our VPG-based parser
  generator is verified in Coq.  Correctness means that if the
  generated parser constructs a parse tree, it must be a valid parse
  tree according to the input VPG, and vice versa. However, 
  there are gaps between our VPG parser
  generator's Coq formalization of and its implementation in OCaml. First,
  the implementation takes tagged CFGs as input and translates
  tagged CFGs to VPGs; this translation algorithm has not been
  formally modeled and verified in Coq. Second, the
  implementation uses efficient data structures for performance, while
  their Coq models use equivalent data structures that are slower but
  easier for reasoning. For example, the OCaml implementation 
  uses hash tables for storing transition tables of the two PDAs to
  have efficient search (with $O(1)$ search complexity), while the Coq
  counterpart uses a balanced tree (with $O(\log(n))$ search
  complexity) provided as a Coq library. These gaps prevent us from
  directly extracting OCaml code from the Coq formalization.}

\secondRound{Our parsing algorithm requires a VPG as the input
  grammar. Compared to a CFG, a VPG requires partitioning terminals
  into plain, call, and return symbols. Some CFGs may not admit such
  kind of partitioning; the same terminal may require different stack
  actions for different input strings. In particular, all languages
  recognized by VPGs belong to the set of deterministic context-free
  languages, which is a strict subset of context-free languages (the
  classic example that separates CFL from DCFL is
  $\predset{a^ib^jc^k}{i \neq j \lor j \lor k}$).
  We plan to extend our preliminary study on ANTLR
  grammars to understand how much of the syntax of practical computer
  languages (e.g., programming languages and file formats) can be
  described by VPGs.

  The parsing in our VPG framework is performed in three steps:
  construction of a parse forest (with possible invalid edges),
  pruning of the parse forest, and extraction of parse trees. This
  design simplifies formal verification but adds a pruning step.  We
  believe the pruning step can be possibly removed by redesigning our
  parser generator, where the recognizer is extended to a transducer
  that generates semantic actions and the execution of those
  semantic actions builds the parse forest directly.}

The translation algorithm from tagged CFGs to VPGs is sound but not
complete. In general, it is an open problem to determine whether a CFG
can be translated to a VPG, and to infer the call and return symbols
automatically.

\section{Conclusions}
In this paper, we present a recognizer and a formally verified parser
generator for visibly pushdown grammars.  The parsing algorithm is
largely enlightened by the recognizer, with several trade-offs to
simplify the structure and reduce the burden of formal verification.
We also provide a surface grammar called tagged CFGs and a translator
from tagged CFGs to VPGs.  We show that when a format can be modeled
by a VPG and its call and return symbols can be identified, VPG
parsing provides competitive performance and sometimes a significant
speed-up.

\section*{Acknowledgment}
The authors would like to thank anonymous reviewers for their
insightful comments. This work was supported by DARPA research grant
HR0011-19-C-0073.

\bibliographystyle{ACM-Reference-Format}
\bibliography{paper}

\appendix

\section{Correctness proofs of the recognizer} \label{app:correct_recognizer}

\begin{lemma}
\label{lemma1}
    If $(S_1,T)\leadsto w \land S_1\subseteq S_2$, then
    $(S_2,T)\leadsto w$.
\end{lemma}
\begin{proof}
    By Definition, if $T=\bot$, then
    $$(S_1,T)\leadsto w\Ra \exists (L_1,L_2)\in S_1,L_2\deriv  w.$$
    Since $S_1\subseteq S_2$, we have $(L_1,L_2)\in S_2$, so $(S_2,T)\leadsto w$.

    Otherwise $T=[S',\ac ]\cdot T'$, then $w=w_1 \bc  w_2 $, and 
              $\exists (L_3,L_4)\in S_1$ s.t.
              \begin{enumerate}
                  \item $L_4\deriv  w_1$ and
                  \item $\exists (L_1,L_2)\in S',\exists L_5, L_2\ra \br{L_3}L_5 \land (\{(L_1,L_5)\},\ T')\leadsto w_2$.
              \end{enumerate}
    Again, since $S_1\subseteq S_2$, we have $(L_3,L_4)\in S_2$,
    so $(S_2,T)\leadsto w$.
\end{proof}

\begin{lemma}
\label{lemma2}
    If $L_2\deriv  w_1L_3$ and $(\{(L_1,L_3)\},T)\leadsto w$,
    then $(\{(L_1,L_2)\},T)\leadsto w_1w$.
\end{lemma}
\begin{proof}
    By definition, if $T=\bot$, then
    $$(\{(L_1,L_3)\},\bot)\leadsto w \mbox{ implies } L_3\deriv  w.$$
    Thus $L_2\deriv  w_1L_3$ implies  $L_2\deriv  w_1w$.
    By definition, $(\{(L_1,L_2)\} \bot)\leadsto w_1w$.

    Otherwise $T=[S',\ac ]\cdot T'$, then $w=w' \bc  w'' $,
              \begin{enumerate}
                  \item $L_3\deriv  w'$ and
                  \item $\exists (L',L'')\in S',\exists L_5, L''\ra \br{L_1}L_5 \land (\{(L',L_5)\},\ T')\leadsto w''$.
              \end{enumerate}

    Thus $L_2\deriv  w_1L_3$ implies $L_2\deriv  w_1w'$.
    By definition, $(\{(L_1,L_2)\},T)\leadsto w_1w$.
\end{proof}

\begin{lemma}
\label{lemma3}
    If $(S,T)\leadsto w$, then $\exists (L_1,L_2)\in S$,
    s.t. $(\{(L_1,L_2)\}, T)\leadsto w$.
\end{lemma}
\begin{proof}
    By definition, if $T=\bot$, then
    $$(S,T)\leadsto w\Ra \exists (L_1,L_2)\in S,L_2\deriv  w.$$
    Then by definition again we have the lemma.

    Otherwise $T=[S',\ac ]\cdot T'$ and $w=w_1 \bc  w_2 $. By definition again we have the lemma.
\end{proof}

\begin{theorem}
\label{thPlain}
    Assume $\delta_c(S)=(S',\lambda T.T)$ for a plain symbol $c$. Then
    $(S,T)\leadsto cw$ iff $(S',T) \leadsto w$.
\end{theorem}
\begin{proof}
    $\Ra$. By case over $(S,T)\leadsto cw$.
    \begin{enumerate}
        \item $T=\bot$ and $(S,\bot)\leadsto cw$. By definition,
              $\exists (L_1,L_2)\in S \land L_2\deriv cw$.
          By the forms of VPG rules, we must have
          \[\exists L_3, L_2\ra cL_3\ \land\ L_3\deriv w.\]
              Since $(L_1,L_2)\in S \land L_2\ra cL_3$, by the definition of $\delta_c$,
              and we have $(L_1,L_3)\in S'$,
              thus since $L_3\deriv w$ we have $(S',\bot)\leadsto w$.
        \item $T=[S_1,\ac ]\cdot T'$ and $(S, [S_1,\ac ]\cdot T')\leadsto cw$. By definition,
              $w=w_1 \bc  w_2 $, and
              $\exists (L_3,L_4)\in S$ s.t.
              \begin{enumerate}
                  \item $L_4 \deriv  cw_1$ and
                  \item $\exists (L_1,L_2)\in S_1,\exists L_5, L_2\ra \br{L_3}L_5 \land (\{(L_1,L_5)\},\ T')\leadsto w_2$.
              \end{enumerate}

              Since $L_4\deriv cw_1$, we have $\exists L_4'$, $L_4\ra cL_4'\ \land\
                  L_4'\deriv w_1$, thus $(L_3,L_4')\in S'$, and
              $(S',T)\leadsto w_1\bc  w_2=w$.
    \end{enumerate}
    $\La$.

    By Lemma~\ref{lemma3}, $$\exists (L_1,L_3)\in S'\ \st\
        (\{(L_1,L_3)\},T)\leadsto w.$$
    Thus, by the definition of derivatives, $\exists L_2, (L_1,L_2)\in S \land L_2\ra cL_3$.
    By Lemma~\ref{lemma2}, $(\{(L_1,L_2)\},T)\leadsto cw$.
    By Lemma~\ref{lemma1}, $(S,T)\leadsto cw$.
\end{proof}

\begin{theorem}
\label{thOpen}
    Assume $\delta_{\ac}S=(S',\lambda T.[S,\ac ]\cdot T)$ for a call symbol $\ac$.
    Then $(S,T)\leadsto \ac w$ iff $(S', [S,\ac ]\cdot T)\leadsto w$.
\end{theorem}

\begin{proof}
    $\Ra$. By case over $(S,T)\leadsto \ac w$.
    \begin{enumerate}
        \item $T=\bot$ and $(S,\bot)\leadsto \ac w$. By definition $\exists (L_1,L_2)\in S \land
                  L_2\deriv  \ac w$. Thus from the forms of VPG rules we have
              $$\exists L_3,L_4,\bc \ \st L_2\ra \brp{a}{b}{L_3}L_4\deriv  \ac w.$$
              Thus $\exists w_1,w_2$,
              s.t. $L_3\deriv w_1\land L_4\deriv w_2\land
                  w=w_1\bc  w_2$.

              By definition, $$(\{(L_1,L_4),\bot\})\leadsto w_2,$$
              thus $(\{(L_3,L_3)\},[S,\ac ]\cdot\bot)\leadsto
                  w_1\bc  w_2$.
              Since $(L_3,L_3)\in\partial_{\ac}S$, we have
              $(\partial_{\ac}S, [S,\ac ]\cdot\bot)\leadsto
                  w_1\bc  w_2=w$.
                  \item $T=[S_1,\langle c]\cdot T'$ and $(S,[S_1,\langle c]\cdot T')\leadsto \ac w$. Then by definition we have $w=w_1d\rangle w_2$, and

              $\exists (L_3,L_4)\in S$ s.t.
              \begin{enumerate}
                  \item $L_4\deriv  \ac w_1$ and
                  \item $\exists (L_1,L_2)\in S_1,\exists L_5, L_2\ra \brp{c}{d}{L_3}L_5 \land (\{(L_1,L_5)\},\ T')\leadsto w_2$.
              \end{enumerate}

              Thus $\exists L_6,L_7$ s.t.
              $L_4\ra\br{L_6}L_7$, $L_6\deriv w_{11}$, $L_7\deriv w_{12}$ and
              $w_1=w_{11}\bc  w_{12}$.
              Thus by definition, $(L_6,L_6)\in\partial_{\ac}S$,
              so
              $$(\partial_{\ac}S, [S,\ac ]\cdot[S_1, \langle c] \cdot T')\leadsto w_{11}\bc  w',$$
              where $w'$ satisfies $(\{(L_3,L_7)\},[S_1, \langle c] \cdot T')\leadsto w'$.
              From $(\{(L_1,L_5)\},T')\leadsto w_2 $, we know $w'$ can be $w_{12}d\rangle w_2$,
              so we have
              $$(\partial_{\ac}S, [S,\ac ]\cdot T)\leadsto w_{11} \bc  w_{12}d\rangle w_2=w.$$
    \end{enumerate}

    $\La$.
    By definition, $\exists(L_3,L_3)\in S'$, s.t. $w=w_1\bc  w_2$, and
    \begin{enumerate}
        \item
              $\exists(L_1,L_2)\in S\ \st L_3\deriv  w_1 \land L_2\ra \brp{a}{b}{L_3}L_5$.
        \item
              $(\{(L_1,L_5)\},T)\leadsto w_2$.
    \end{enumerate}

    Since $L_2\ra\brp{a}{b}{L_3}L_5 \land
        (\{(L_1,L_5)\},T)\leadsto w_2$, by Lemma~\ref{lemma2}, we have
    $(\{L_1,L_2\},T)\leadsto\brp{a}{b}{w_1}w_2=\ac w$, thus
    $(S,T)\leadsto\ac w$.
\end{proof}

\begin{theorem}
\label{thClose}
    Assume $\delta_{\bc}(S, [S_1,\ac ])=(S',\tail)$ for a return symbol $\bc$.
    Then $(S, [S_1,\ac ]\cdot T)\leadsto \bc  w$ iff
    $(S',T)\leadsto w$.
\end{theorem}
\begin{proof}
    $\Ra$. By definition, $\exists(L_3,L_4)\in S$ and $(L_1,L_2)\in S_1$,
    and we must have $L_4\deriv\epsilon$, $L_2\ra\brp{a}{b}{L_3}L_5$
    and $(\{(L_1,L_5)\},T)\leadsto w$. From $L_4\deriv\epsilon$ and the
    forms of allowed VPG rules, we must have $L_4\ra\epsilon$; therefore,
    $(L_1,L_5)\in S'$.
    By Lemma~\ref{lemma1}, we have $(S',T)\leadsto w$.

    $\La$. By Lemma~\ref{lemma3}, we have
    $(L_1,L_5)\in S' \land
        (\{(L_1,L_5)\},T)\leadsto w$,
    thus $\exists (L_1,L_2)\in S_1 \land (L_3,L_4)\in S \land
        L_4\ra\epsilon \land L_2\ra \brp{a}{b}{L_3}L_5$,
    by definition, we have $(S,[S_1,\ac ]\cdot T)\leadsto \bc  w$.
\end{proof}

We next define the runtime execution of recognizer PDAs. It is standard and we include it here so that we can state the correctness theorem formally.
Recall that a {runtime configuration} of a recognizer PDA is a pair $(S,T)$,
    where $S$ is a state and $T$ is a stack, denoted as $T=t_1\cdot t_2\cdots t_k\cdot\bot$, where $t_i=[S_i, \ac _i]$ for $i=1..k$,
    $\ac_i \in\Sigma_c$ is a call symbol, $t_1$ the top of $T$,
  and $\bot$ the empty stack. The PDA's initial configuration is
$\set{(L_0,L_0)},\bot)$ and its acceptance configurations are defined as follows.
\begin{definition}[PDA acceptance configurations]
    Given a VPG $G=(V,\Sigma,P,L_0)$, pair $(S,T)$ is an \emph{acceptance configuration} if $T=\bot$, and
    $\exists (L,L')\in S\ \st (L'\ra\epsilon)\in P$.
\end{definition}

\newcommand{\RTranstrans}{\mathcal{F}}

\begin{definition}[Recognizer PDA execution]
    The runtime execution $\RTranstrans$ of a PDA $(S_0,A,\RTrans)$ is defined as follows, where $S_0$ is the start state, $A$ the set of states, and $\RTrans$ the set of configuration transitions.
    \begin{equation*}
        \RTranstrans:(i,S,T)\mapsto (S',T'),
    \end{equation*}
    where
    \begin{enumerate}
        \item
              if $i\in\Sigma_c\cup\Sigma_l$, then $(S,S')\in \RTrans$ and is marked with $(i,f)$, and $T'=f(T)$;
        \item
              if $i\in\Sigma_r$ and $T=t\cdot T'$, then $(S,S')\in \RTrans$ and is marked with $(i,t,f)$, and $T'=f(T)$.
    \end{enumerate}
\end{definition}

Given an input string $w = w_1\ldots w_n$, we say PDA accepts $w$ if there exists
a sequence of configurations $(S_0, T_0), \ldots, (S_n, T_n)$ so that 
\[
\begin{array}{l}
            (S_0,T_0)      =(\{(L_0,L_0)\},\bot),    \\
            (S_{i},T_{i})  =\RTranstrans (w_i,S_{i-1},T_{i-1}), \mbox{for } i\in[1,n] \\
            (S_n,T_n)\ \mbox{is an acceptance configuration}
\end{array}
\]
Otherwise, $w$ is rejected.

\begin{lemma}
\label{tfacc}
    Given a VPG $G=(V,\Sigma,P,L_0)$, suppose a PDA is generated according to Algorithm~\ref{alg:constr_recog}. Then for a string $w\in\Sigma^*$, $(\{(L_0,L_0)\},\bot)\leadsto w$ iff $w$ is accepted by the PDA.
\end{lemma}
\begin{proof}
    $\La.$ If $w$ of length $k$ is accepted by the PDA, then there exists 
  a sequence of configurations $(S_i,T_i), i\in[0..k]$, s.t. (1)
  $S_0=\{(L_0,L_0)\}$, $T_0=\bot$; (2) 
  $(S_{i},T_{i})=\RTranstrans(w_{i},S_{i-1},T_{i-1}),\ i \in [1..k]$,  where
          $w_i$ is the $i$-th symbol in $w$ and $\RTranstrans$ is the PDA transition function; and (3)
  $(S_k,T_k)$ is an acceptance configuration.

  For each $i$, perform case analysis over $w_i$.  Suppose $w_i$ is a
  plain symbol, denoted as $c$. By
  $(S_{i},T_{i})=\RTranstrans(w_{i},S_{i-1},T_{i-1})$ and the PDA
  construction, we must have $\delta_c(S_{i-1}) = (S_i, \lambda T. T)$.
  By Theorem~\ref{thPlain}, we get 
    $(S_{i-1},T_{i-1})\leadsto w_i \ldots w_k$ iff 
    $(S_i,T_i) \leadsto w_{i+1} \ldots w_k$.
  The cases for when $w_i$ is $\ac$ or $\bc$ are similar, with the help
  of Theorems~\ref{thOpen} and~\ref{thClose}. 

  Combining all steps, we have $(S_0,T_0)\leadsto w$ iff
  $(S_k,T_k) \leadsto \epsilon$.  Since $(S_k,T_k)$ is an acceptance
  configuration, we have $(S_k,T_k) \leadsto \epsilon$.  Therefore,
  we get $(S_0,T_0)\leadsto w$.

$\Ra.$ We prove a more general lemma: if $(S_0,T_0)$ is a PDA runtime
  configuration and $(S_0,T_0) \leadsto w$, then $w$ is accepted by the PDA.
Prove it by induction over the length of $w$. 

When the length is zero, we must have $T_0=\bot$ and there exists
$(L_1,L_2) \in S_0$ such that $L_0 \ra \epsilon$. Therefore
$(S_0,T_0)$ is an acceptance configuration of the PDA.

For the inductive case, suppose $w=w_1 \ldots w_{k+1}$. Perform case
analysis over $w_1$, and first show that there exists $(S_1,T_1)$ s.t.
$(S_1,T_1)=\RTranstrans(w_1,S_0,T_0)$.
\begin{enumerate}
\item Suppose $w_1$ is a plain symbol $c$ and
  $\delta_c(S_0)=(S_1,\lambda T, T)$.  Since Algorithm~\ref{alg:constr_recog} is
  closed under derivatives, we have $S_1$ is a PDA state. Let
  $T_1=T_0$. Thus, $\RTranstrans(w_1,S_0,T_0) = (S_1,T_1)$ by the
  definition of $\RTranstrans$.

\item The case of $w_1$ being a call symbol is similarly to the previous case.

\item Suppose $w_1$ is a return symbol $\bc$.  By $(S_0,T_0) \leadsto
  \bc w_2 \ldots w_{k+1}$, we have $T_0$ is not the empty stack and has
  a top symbol $[S,\ac]$. Suppose 
  $\delta_{\bc}(S_0, [S,\ac])=(S_1,\tail)$. Since Algorithm~\ref{alg:constr_recog} is
  closed under derivatives, we have $S_1$ is a PDA state. Let
  $T_1=\tail(T_0)$. Thus, $\RTranstrans(w_1,S_0,T_0) = (S_1,T_1)$ by the
  definition of $\RTranstrans$.
\end{enumerate}
By Theorems~\ref{thPlain},~\ref{thOpen}, and~\ref{thClose}, we get
$(S_1,T_1) \leadsto w_2 \ldots w_{k+1}$. By the induction hypothesis, $w_2
\ldots w_{k+1}$ is accepted by the PDA. Therefore, the original string
$w_1 \ldots w_{k+1}$ is also accepted.

\end{proof}

\begin{theorem}
    For VPG $G$ and its start symbol $L_0$, a string $w\in\Sigma^*$ is derived from $L_0$, i.e.
    $L_0\deriv w$, iff $w$ is accepted by the corresponding PDA.
\end{theorem}

\begin{proof}
    By Lemma~\ref{tfacc}, $w$ is accepted by the PDA iff $(\{(L_0,L_0)\}, \bot)\leadsto w$, and by definition we have $(\{(L_0,L_0)\}, \bot)\leadsto w$ iff $L_0\deriv w$.
\end{proof}

\section{Recognizing strings with pending calls/returns}\label{app:recog_general_VPG}

In this section, we extend the work in Section~\ref{sec:recognition} to build PDAs for recognizing VPG with pending call or return symbols.
In general VPGs, nonterminals are classified to
two categories: $V^0$ for matching well-matched strings and $V^1$ for
strings with pending calls/returns.  We write $V=V^0\cup V^1$ for the
set of all nonterminals. $V^0$ should be disjoint from $V^1$.
The definition also imposes constraints on how $V^0$ and $V^1$ nonterminals
can be used. E.g., in $L\ra \br{L_1}L_2$, $L_1$ must be in a well-matched
nonterminal (i.e., in $V^0$). This constraint excludes a grammar like
$L_1 \ra \br{L_2}L_3; L_2 \ra \callsym{c} L_4$.

Another major difference is that in $L\ra a L_1$, symbol $a$ can be a
call/return symbol in addition to being a plain symbol. This makes
matching calls and returns more complicated. For example, suppose we
have rules: $L_1 \ra \ac L_2$;\ $L_2 \ra \bc L_3 | \epsilon$;\ $L_3 \ra
\epsilon$.  Then string $\ac \bc$ is accepted, in which case $\bc$
from $L_2 \ra \bc L_3$ matches $\ac$ from $L_1 \ra \ac
L_2$.  String $\ac$ is also accepted, in which case $\ac$ is a pending
call. So depending on the input string, $\ac$ from $L_1 \ra \ac
L_2$ may be a matching call or a pending call.

Here's an example grammar:
\begin{enumerate}
    \item $\derivRule{L_1}{\ac L_2 \bc L_3}$
    \item $\derivRule{L_2}{\ac L_2 \bc L_4\mid \epsilon}$
    \item $\derivRule{L_3}{\ac L_1\mid \epsilon}$
    \item $\derivRule{L_4}{\epsilon}$
\end{enumerate}
And $L_1, L_3 \in V^1, L_2, L_4 \in V^0$.
For example, $\ac \bc \ac \ac \ac \bc \bc$ is in the language recognized by the grammar.

\paragraph{General VPGs to PDA}
The PDA states and stack symbols are the same as
before. We generalize the notion of the top of the stack to return the
top stack symbol when the stack is non-empty, and return None when the
stack is empty. 

A derivative-based transition function takes the current state and the top of the stack
(which can be None), and returns a new state and a stack action. 
As before, since $\delta_c$ and $\delta_{\ac}$ do
not use the top of the stack, we omit it from their parameters.

\begin{definition}[Derivative functions for general VPGs]
    
    Given a general VPG $G=(V,\Sigma,P,L_0)$, the transition functions $\trans$ are defined as follows.
    For $c\in\Sigma_l$, $\ac\in\Sigma_c$ and $\bc\in\Sigma_r$,
    \begin{enumerate}
        \item
              $\delta_c$ is the same as the well-matched case.

              $\delta_c(S) = (S',\lambda T.T)$, where
              $$
                  S'=\{(L_1,L_3)\mid \exists L_2, (L_1,L_2)\in S \land (L_2\ra cL_3)\in P\};
              $$

              \item For call symbols, we have
              $\delta_{\ac}(S)= (S'\cup S_p,\lambda T.[S,\ac ]\cdot T)$, where
          \[
              \begin{array}{lll}
                S' & = \{(L_3,L_3)\mid \exists L_1\; L_2, (L_1,L_2)\in S\ \land
                     \exists L_4, (L_2\ra \brp{a}{b}{L_3}L_4)\in P \}, \\
                S_p & = \{(L_3,L_3)\mid \exists L_1,L_2,\ (L_1,L_2) \in S\ \land
                 (L_2\ra \ac{L_3}) \in P \}.
              \end{array}
         \]
         Compared to the well-matched case, an additional $S_p$ is
         introduced for the case when $\ac$ appears in a rule like
         $L_2\ra \ac{L_3}$.


        \item For a return symbol $\bc$, if $t$ is the top of the stack, then
          \[\delta_{\bc}(S,t) = 
          \left\{
               \begin{array}{ll}
                 (S' \cup S_{p1}, \tail) & \mbox{if}\ t= [S_1,\ac]\\
                 (S_{p2}, \lambda T. T) & \mbox{if}\ t= \mbox{None}
               \end{array}
          \right.
          \]
         where
         \[
            \begin{array}{lll}
                S' & = &
                \{(L_1,L_5)\mid \exists L_2\; L_3\; L_4, (L_1,L_2)\in S_1 \land
                  (L_3,L_4)\in S\; \land \\
                  &   & ~~~~~~~~~~~~~~~(L_4\ra\epsilon)\in P\ \land (L_2\ra\brp{a}{b}{L_3}L_5)\in P \} \\
                S_{p1} & = &
                \{(L_1,L_5)\mid \exists L_2\; L_3\; L_4, (L_1,L_2)\in S_1 \land
                  (L_3,L_4)\in S\; \land \\
                  &   & ~~~~~~~~~~~~~~~(L_2\ra\ac L_3)\in P\ \land (L_4\ra \bc L_5)\in P \} \\
                S_{p2} & = &
                \{(L_3,L_3)\mid \exists L_1\; L_2, (L_1,L_2)\in S \land (L_2\ra \bc L_3)\in P \} \\
              \end{array}
         \]

         $S'$ is as before and deals with the case when there is a
         rule $L_2\ra\brp{a}{b}{L_3}L_5$ with a proper top stack
         symbol.  $S_{p1}$ deals with the case when there are rules
         $L_2\ra\ac L_3$ and $L_4\ra \bc L_5$; in this case, we match
         $\bc$ with $\ac$. Finally, $S_{p2}$ deals with the case when
         the stack is empty; then $\bc$ is treated as a pending return
         symbol (not matched with a call symbol).

    \end{enumerate}
\end{definition}

For the well-matched case, the stack should be empty after all
input symbols are consumed; in the case with pending calls/returns,
however, the stack is not necessarily empty at the end. For example,
with the grammar $L\ra \ac L\mid \epsilon$ and the valid input string
$\ac$, the terminal stack is $[\{(L,L)\},\ac]\cdot\bot$. 

\begin{definition}[The acceptance configuration for words with pending calls/returns] 
    Given a general VPG $G=(V,\Sigma,P,L_0)$, the pair $(S,T)$ is called an \emph{acceptance configuration} if the followings are satisfied:
\begin{enumerate}
    \item $\exists (L_1,L_2) \in S \text{ s.t. } (L_2 \ra \epsilon) \in P$,
    \item either (i) $T=\bot$ or (ii) $T= [S', \ac] \cdot T'$ and $\exists (L_3,L_4) \in S' \land (L_4 \ra \ac L_1) \in P$ for some $L_1$.
    \end{enumerate}
\end{definition}

\newcommand{\mret}[1]{\operatorname{matched-rets}(#1)} 
\newcommand{\wmatched}[1]{\operatorname{well-matched}(#1)}

In the following correctness proof,
we use predicate $\wmatched{w}$ to mean that $w$, a string of
terminals, is a well-matched string; that is, every call/return symbol
is matched with a corresponding return/call symbol. We use predicate
$\mret{w}$ to mean that any return symbol in $w$ is matched with a
call symbol; however, a call symbol may not be matched with a return
symbol. E.g., we have $\mret{\ac \ac \bc}$, but not $\wmatched{\ac \ac \bc}$.

\begin{definition}[Semantics of PDA configurations]
    We will write $(S, T)\leadsto w$ to mean that $w$ can be
    accepted by the configuration $(S, T)$. It is defined as follows.
    \begin{enumerate}
        \item $(S, \bot)\leadsto w \mbox{ if } \exists (L_1,L_2)\in S,\ \st L_2\deriv w$,
        \item $(S, [S',\ac ]\cdot T')\leadsto w_1 \bc  w_2$ if
              $\exists (L_3,L_4)\in S$ s.t.
              \begin{enumerate}
                  \item $L_4\deriv  w_1$ and $\wmatched{w_1}$ and
                  \item $\exists (L_1,L_2)\in S',\exists L_5, L_2\ra \br{L_3}L_5 \land (\{(L_1,L_5)\},\ T')\leadsto w_2$.
              \end{enumerate}
        \item $(S, [S',\ac ]\cdot T')\leadsto w_1 \bc  w_2$ if
              $\exists (L_3,L_4)\in S$ s.t. $\exists L_5$
              \begin{enumerate}
                  \item $L_4\deriv  w_1 \bc L_5$ and $\wmatched{w_1}$ and
                  \item $\exists (L_1,L_2)\in S',L_2\ra \ac L_3 \land (\{(L_1,L_5)\},\ T')\leadsto w_2$.
              \end{enumerate}
        \item $(S, [S',\ac ]\cdot T')\leadsto w_1$ if
              $\exists (L_3,L_4)\in S$ s.t.
              \begin{enumerate}
                  \item $L_4\deriv  w_1$ and $\mret{w_1}$ 
                  \item $\exists (L_1,L_2)\in S',L_2\ra \ac L_3$.
              \end{enumerate}
    \end{enumerate}
\end{definition}
In the above definition, the third case handles when the
call symbol $\ac$ in rule $L_2\ra \ac L_3$ matches $\bc$ in $w_1 \bc
L_5$ produced by $L_4$. The last case handles when $\ac$ in rule
$L_2\ra \ac L_3$ does not have a matched return; that is, it is a
pending call.

The following three lemmas and their proofs are the same as
before (except that Lemma~\ref{lemma2'} requires well-matched
strings).

\begin{lemma}
    \label{lemma1'}
    If $(S_1,T)\leadsto w \land S_1\subseteq S_2$,
    then $(S_2,T)\leadsto w$.
\end{lemma}

\begin{lemma}
    \label{lemma2'}
    If $L_2\deriv  w_1L_3$, $\wmatched{w_1}$, and $(\{(L_1,L_3)\},T)\leadsto w$,
    then $(\{(L_1,L_2)\},T)\leadsto w_1w$.
\end{lemma}

\begin{lemma}
    \label{lemma3'}
    If $(S,T)\leadsto w$, then $\exists (L_1,L_2)\in S$,
    s.t. $(\{(L_1,L_2)\}, T)\leadsto w$.
\end{lemma}

In addition, we need the following lemma.
\begin{lemma}
    \label{lemma4'}
    If $L \deriv w \delta$, where $\delta\in (\Sigma\cup V)^*$ is a string of terminals or nonterminals, then either (1) $\mret{w}$, or (2) exists $w_1, \bc, w_2$, so that $w=w_1 \bc w_2$ and $\wmatched{w_1}$ and exists $L_1$ so that $L \deriv w_1 \bc L_1$ and $L_1 \deriv w_2 \delta$.
\end{lemma}
\begin{proof}
Sketch: If $w=\epsilon$, then $\mret{\epsilon}$. Otherwise,
prove it by induction over the length of the derivation of $L \deriv w
\delta$, and then perform case analysis over the first derivation step.
\end{proof}

\begin{theorem}
    \label{thPlain'}
    For a plain symbol $c$,
    $(S,T)\leadsto cw$ iff $\delta_c(S)=(S',f)$, and $(S',f(T)) \leadsto w$.
\end{theorem}
The proof is similar to the proof before, except the $\Ra$ direction
has more cases to consider.

\begin{theorem}
    \label{thOpen'}
    For $\ac\in\Sigma_c$,
    $(S,T)\leadsto \ac w$ iff
    $\delta_{\ac}S=(S',f)$, and $(S', f(T))\leadsto w$.
\end{theorem}
The proof is similar to the proof before, except with more cases to
consider. The $\Ra$ direction requires the use of Lemma~\ref{lemma4'}.

\begin{theorem}
   \label{thClose'}

   \begin{enumerate}[(1)]
     \item If $\delta_{\bc}(S, [S_1,\ac ])=(S',\tail)$, then $(S, [S_1,\ac ]\cdot T)\leadsto \bc  w$ iff  $(S',T)\leadsto w$.
     \item If $\delta_{\bc}(S,\none)=(S',\lambda T. T)$, then $(S, \bot)\leadsto \bc  w$ iff $(S',\bot)\leadsto w$.
   \end{enumerate}
\end{theorem}
Part (1)'s proof is similar to before, except with more cases and sometimes
need to use Lemma~\ref{lemma4'}. Part(2)'s proof is straightforward.

\begin{algorithm}
    \caption{PDA construction}
    \label{getpda_unmatch}
    \begin{algorithmic}[1]
        \STATE Input: a VPG $G=(V,\Sigma,P,L_0)$, $\trans$.
        \STATE $S_0\la \{(L_0,L_0)\}$.
        \STATE Initialize the new state set $N=\{S_0\}$.
        \STATE Initialize the set for all produced states $A=N$.
        \STATE Initialize the set for transitions $\RTrans=\{\}$.
        \REPEAT
        \STATE $N'\la \{ (i,f,S,S')\mid (S',f)=\delta_{i}S, S\in N, i\in\Sigma_c\cup\Sigma_l \}$
        \STATE {Add edge $(S,S')$ marked with $(i,f)$ to $\RTrans$, where $(i,f,S,S')\in N'$}.
        \STATE $R\la \{ [S, \ac]\mid S\in A, \ac\in\Sigma_a \}$
        \STATE $N_R\la\{ (\bc,r,f,S,S')\mid (S',f)=\delta_{\bc}(S, r), S\in A, \bc\in\Sigma_r, r\in R \cup \set{\none} \}$
        \STATE {Add edge $(S,S')$ marked with $(\bc,r,f)$ to $\RTrans$, where $(\bc,r,f,S,S')\in N_R$}.
        \STATE {$N\la\{S'\mid (\_,\_,\_,S')\in N' \lor  (\_,\_,\_,\_,S')\in N_R\}-A$}
        \STATE $A\la A\cup N$
        \UNTIL{$N=\emptyset$}
        \STATE Return $(S_0,A,\RTrans)$.
    \end{algorithmic}
\end{algorithm}

\begin{lemma}
    \label{tfacc'}
    Given a VPG $G=(V,\Sigma,P,L_0)$, suppose a PDA is generated according to Algorithm~\ref{getpda_unmatch}. Then for a string $w\in\Sigma^*$, $(\{(L_0,L_0)\},\bot)\leadsto w$ iff $w$ is accepted by the PDA.
\end{lemma}
The lemma can be proved as before, except with more cases.

\begin{theorem}
    \label{corecog'}
    For VPG $G$ and its start nonterminal $L_0$, a string $w\in\Sigma^*$ is derived from $L_0$, i.e.
    $L_0\deriv w$, iff $w$ is accepted by the corresponding PDA.
\end{theorem}
The proof is as before.

\section{The pruner and the extractor} \label{app:prunerExtractor}

\subsection{The pruner}

First we define the pruning function for the last state
$m_n$ in the reversed parse tree $[m_1, \ldots, m_n]$.  
Only the edges that end with $L_1^\false$ for some $L_1$ so that
$(L_1\ra \epsilon)\in P$ is valid. This reflects how parsing 
is finished successfully: (1) $L_1$ tagged with false indicates no
matching rule is waiting to be finished; and (2) the parsing
can end due to $L_1\ra \epsilon$.

\begin{definition}[The pruning function for the last state]
    The function that prunes the last state, denoted as $g_\epsilon$, is defined as 
    $$g_\epsilon(m_n)=\{(\_,\_, L_1^\false)\in m_n\mid (L_1\ra \epsilon)\in P\}.$$
\end{definition}

After pruning the last state, we continue pruning previous
states. Assume we have the current pruner configuration $(m_2', T)$,
where $m_2'$ has already been pruned, and want to prune the previous
state $m_1$ from the input parse forest. The pruning function $g$ prunes
$m_1$ and transitions to a new configuration $(m_1', T_1)$, where
$m_1'$ is the pruned state of $m_1$. 

\begin{definition}[Transition functions for the pruner PDA]
    Given a current state $m_2'$, a stack $T$,  and input parse-forest state $m_1$, 
    the transition functions $g$ are defined as follows.

    \begin{enumerate}
    \item $g(m_1,m_2')=(m_1',\lambda T.T)$, where $(m_1,m_2')\subseteq (\mml\times\mml)\cup (\mml \times \mmc)\cup (\mmr \times \mml)\cup (\mmr \times \mmc)$, and
        $$m_1'=\{(\_, \_, L_1^u)\in m_1\mid (L_1^u,\_,\_)\in m_2'\}.$$
        
        For an edge in $m_1$ to be valid, it must connect to one edge in $m_2'$; otherwise, the edge cannot be in a trace that continues to $m_2'$ and can be pruned. The stack is not modified in this case.
    
    \item $g(m_1,m_2')=(m_1',\lambda T.\; m_2'\cdot T)$, where 
        $(m_1,m_2')\subseteq (\mml\times\mmr)\cup (\mmr \times \mmr)$, and
        $$
            m_1'=
            \{(\_,\_, L_1^\false)\in m_1\mid \exists(L_1^\false,\_,\_)\in m_2'\}\cup 
            \{(\_,\_, L_1^\true)\in m_1\mid (L_1\ra\epsilon)\in P\}.
        $$
        
        In this case, consider an edge $e=(\_,\_,L_1^u)\in m_1$.  If
        $u=\false$, then $e$ must connect to $m_2'$; this is similar
        to the above case.  If $u=\true$, then $(L_1\ra\epsilon) \in
        P$ so that the next edge can be a return edge.  The 
        state $m_2'$ must be pushed to the stack. To see the reason,
        assume the parser PDA applied the rule $L\ra\br{A}E$ and
        generated $(L,\ac,A)$ and $((L,A),\bc,E)$. How do we know if
        $(L,\ac,A)$ is valid?  If $((L,A),\bc,E)$ does not exist when
        we prune the state that contains $(L,\ac,A)$, then $(L,\ac,A)$
        can be pruned. Therefore, $m_2'$ is pushed to the stack for
        later use when pruning the corresponding $\mc$.

    \item $g(m_1,m_2',{\mr})=(m_1',\tail)$,
         where 
            $(m_1,m_2')\subseteq (\mmc\times\mml)\cup (\mmc \times \mmc)$, 
            and ${\mr}=\hd T$ when $T\neq \bot$ and ${\mr}=\emptyset$ when $T=\bot$, and
        \[
        \begin{array}{l}
            m_1'=\{(\_,\_, L_1^\false)\in m_1\mid \exists(L_1^\false,\_,\_ )\in m_2'\}\ \cup   \\
            \hspace{5ex} \{(L^u,\_,L_1^{\true})\in m_1\ |\
                \exists(L_1^{\true},\_,\_)\in m_2',
                ((L^u,L_1^{\true}),\_,\_)\in{\mr}\}.
        \end{array}
        \]
        As discussed above, we pop $\mr$ from the stack, so that we
        know which return edges those call edges in $m_1$ may connect
        to. If the stack is empty, only pending call edges may be
        valid. Otherwise, we must also consider matching call
        edges. For a matching call edge $e=(L^u,\ac,L_1^\true)\in m_1$,
        there are two restrictions: (1) there must be a return edge in
        $\mr$ that connects to $e$; and (2) there must be an edge in
        $m_2'$ that connects to $e$. Intuitively, after rule
        $L\ra\br{L_1}L_2$ there are two branches, one starts with
        $L_1$, and the other starts with $L_2$; a valid trace must
        cover both branches.

    \item $g(m_1,m_2')=(m_1',\lambda T. T)$, where $(m_1,m_2')\subseteq (\mmc\times\mmr)$, and
        \[
        \begin{array}{l}
            m_1'=
            \{(\_,\_,L^{\false})\in m_1\mid \exists(L^{\false}\_,\_)\in m_2'\}\ \cup \\
            \hspace{5ex} \{(L^u,\_,L_1^{\true})\in m_1\ |\
                \exists((L^u,L_1^{\true})),\_,\_)\in m_2'\land(L_1\ra\epsilon)\in P\}.
        \end{array}
        \]
        The major difference from the above case is that the edge
        $e=(\_,\_,L_1^\true)\in m_1$ must satisfy $L_1\ra\epsilon$ so that
        the parsing of the nested inner string, which is empty, can
        terminate.

\end{enumerate}
\end{definition}

\paragraph{Constructing the pruner PDA}
We use $G_0$ for the set of states in the parser PDA; they are not
pruned. We use $G_1$ for the union of $G_0$ and pruned states from
$G_0$. The pruner PDA construction is the least solution of the
following equation
\[
\begin{array}{lll}
 G_1 & = &
     G_0 \cup G_1 \cup~  \predset{g_\epsilon(m)}{m\in G_0}  \\
  &  & \cup~ \predset{g(m,m')}{m\in G_0 \land m'\in G_1} \\
  &  & \cup~ \predset{g(m,m',m'')}{m\in G_0 \land m'\in G_1 \land m''\in (G_1\cap\Pow(\mmr))\cup\{\emptyset\}} \\
\end{array}
\]

Algorithm~\ref{alg:constr_gPDA} presents an algorithm for the pruner PDA construction.

\begin{algorithm}
    \caption{Constructing the pruner PDA}
    \label{alg:constr_gPDA}
    \begin{algorithmic}[1]
        \STATE Input: the set $A$ of all states in the parser PDA, $g_{\epsilon}$, $g$.
        \STATE $G_0\la A$.
        \STATE $G_1\la G_0\cup\{g_\epsilon(m)\mid m\in G_0\}$.
        \STATE Initialize the set for new states $N=G_1$.
        \STATE Initialize the set for transitions $\gPDA=\{\}$.
        \REPEAT
        \STATE $N'\la \{ (m_2',f,m_1,m_1')\mid (m_1',f)=g(m_1,m_2'),\ m_1\in G_0, \mbox{ and } m_2'\in N\}$
        \STATE {Add edge $(m_2',m_1')$ marked with $(m_1,f)$ to $\gPDA$,
            where $(m_2',f,m_1,m_1')\in N'$}.
        \STATE $N_R\la \{ (m_2',m_1'',f,m_1,m_1')\mid (m_1',f)=g(m_1,m_2',m_1''),\ m_1\in G_0,\ m_2'\in G_1, \mbox{ and } m_1''\in (G_1\cap\Pow(\mmr))\cup\{\emptyset\}\}$
        \STATE {Add edge $(m_2',m_1')$ marked with $(m_1,m_1'',f)$ to $\gPDA$,
            where $(m_2',m_1'',f,m_1,m_1')\in N_R$}.
        \STATE {$N\la\{m_1'\mid (\_,\_,\_,m_1')\in N' \lor  (\_,\_,\_,\_,m_1')\in N_R\}-G_1$}
        \STATE $G_1\la G_1\cup N$
        \UNTIL{$N=\emptyset$}
        \STATE Return $(G_1,\gPDA)$.
    \end{algorithmic}
\end{algorithm}

Given a parse forest with last state $m_n$, the pruner PDA
starts from $(g_\epsilon(m_n),\bot)$ and transitions as follows. 
\begin{definition}[Runtime transition for the pruner PDA]
\label{def:run_gPDA}
    The \emph{runtime transition} $\gPDAtrans$ of the pruner PDA is defined as
    $
        \gPDAtrans(m_1,m_2',T) = (m_1',T'),
    $
    where
    \begin{enumerate}
        \item
              if $(m_1,m_2')\subseteq (\mmc\times\mml)\cup (\mmc \times \mmc)$ and $T=t\cdot T'$, then  $(m_2',m_1')\in \gPDA$ and is marked with $(m_1,t,\tail)$;
        \item
              if $(m_1,m_2')\subseteq (\mmc\times\mml)\cup (\mmc \times \mmc)$ and $T=\bot$, then $(m_2',m_1')\in \gPDA$ and is marked with $(m_1,\emptyset,f)$, and $T'=\bot$;
        \item if $(m_1,m_2')$ belongs to other configurations, then $(m_2',m_1')\in \gPDA$ and is marked with $(m_1,f)$, and $T'=f(T)$.
    \end{enumerate}
\end{definition}

\subsection{The extractor}
\begin{definition}[Connected parse trees]
    For two parse trees $v_1$ and $v_2$, $v_1$ can be \emph{connected to} $v_2$, denoted as $v_1\connect v_2$, is defined as
    $\exists L,u$, \st 
    $$v_1=\_+[(\_,\_,L^u)] \land v_2=[(L^u,\_,\_)]+\_,\text{ or}$$
    $$v_1=\_+[(L^u,\_,L_1^{u_1})] \land v_2=[((L^u,L_1^{u_1}),\_,\_)]+\_.$$
\end{definition}

We use terminology \emph{parse-tree sets} for a set of parse trees
together with corresponding stacks of call edges. We use symbol $V$
for a parse-tree set. Below we define some helper functions for
extracting parse-tree sets from a parse forest $[m_1, \ldots, m_n]$.

\begin{definition}[The helper function $F_1$]
    Given a parse forest $[m_1, \ldots, m_n]$, the helper function $F_1$ converts its first state $m_1$ to the initial parse-tree set $V$.
    \begin{enumerate}
        \item If $m_1 \in \mml$, then $F_1(m_1)=\{([e], \bot) \mid e\in m_1, e=(\_,\_,L^\false) \}$.
        \item If $m_1 \in \mmc$, then $F_1(m_1)=\{([e], e\cdot \bot) \mid e\in m_1, e=(\_,\_,L^u)\}$.
        \item If $m_1 \in \mmr$, then $F_1(m_1)=\{([e], \bot) \mid e\in m_1, e=(\_,\_,L^\false) \}$.
    \end{enumerate}
\end{definition}

\begin{definition}[The helper function $F$]
    Given a parse forest $[m_1, \ldots, m_n]$, the helper function $F$ extends a parse-tree set $V$ based on a state $m$ in the forest to a new parse-tree set $V'$.
    \begin{enumerate}
        \item If  $m \in \mml$ or  $\in \mmc$, $F$ finds the edges $e$ that are connected to the current parse-tree set.
        \begin{align*}
        F(V,m)&=\{(v+[e], E) \mid (v,E)\in V, e\in m \in \mml, v\connect[e] \}, \\
        F(V,m)&=\{(v+[e], e\cdot E) \mid (v,E)\in V, e\in m \in \mmc, v\connect[e] \}.
        \end{align*}
        \item If  $m \in \mmr$, $F$ in addition finds the edges that are connected to the last call edges of the current parse-tree set.
        \[
        \begin{array}{l}
        F(V,m)=\{(v+[e], E') \mid (v,e'\cdot E')\in V \lor ((v,\bot)\in V\land E'=\bot), e\in m, v\connect[e] \} \cup \\
        \hspace{5ex} \{(v+[e], E') \mid (v,e'\cdot E')\in V, v=(\_,\_,L^\true), L\ra\epsilon, e\in m, [e']\connect [e] \}.
        \end{array}
        \]
    \end{enumerate}
\end{definition}

\begin{definition}[The extraction function]
\label{def:extractor}
    Given a parse forest $[m_1, \ldots, m_n]$, we convert it to a parse-tree set in the following way:
    \begin{enumerate}
        \item $V_1=F_1(m_1)$,
        \item $V_i=F(V_{i-1},m_i),\ i=2..n$.
        \item $\extractor{[m_1, \ldots, m_n]}=V_n$.
    \end{enumerate}
\end{definition}

\section{Examples of Tagged CFGs}\label{app:Examples}
The following is a tagged CFG for JSON\footnote{\url{https://github.com/antlr/grammars-v4/tree/master/json/}},
where nonterminals start with lowercase characters, such as ``json'', 
and terminals start with uppercase characters, such as ``STRING''.
Also, call and return symbols are tagged with ``<'' or ``>'', respectively. 
The declarations of terminals are omitted.
\begin{verbatim}
json = value ;
obj = <'{' pair (',' pair)* '}'> | <'{' '}'> ;
pair = STRING ':' value ;
arr = <'[' value (',' value)* ']'> | <'[' ']'> ;
value = STRING | NUMBER | obj | arr | 'true' | 'false' | 'null' ;
\end{verbatim}

The following is a tagged CFG for HTML, 
which is adapted from the HTML grammar 
from the repository of ANTLR\footnote{\url{https://github.com/antlr/grammars-v4/tree/master/html/}}. 
In the following grammar, 
regular operators such as ``?'' or ``*'' are supported by our translator 
(operator ``+'' is also supported but not used here).
\begin{center}
\begin{verbatim}
htmlDocument = scriptletOrSeaWs* XML? scriptletOrSeaWs* DTD?  
    scriptletOrSeaWs* htmlElements* ;
scriptletOrSeaWs = SCRIPTLET | SEA_WS ;
htmlElements = htmlMisc* htmlElement htmlMisc* ;
htmlElement = TagOpen | <TagOpen htmlContent TagClose> 
    | TagSingle | SCRIPTLET | script | style ;
htmlContent = htmlChardata? 
    ((htmlElement | CDATA | htmlComment) htmlChardata?)* ;
htmlAttribute = TAG_NAME (TAG_EQUALS ATTVALUE_VALUE)? ;
htmlChardata = HTML_TEXT | SEA_WS ;
htmlMisc = htmlComment | SEA_WS ;
htmlComment = HTML_COMMENT | HTML_CONDITIONAL_COMMENT ;
script = SCRIPT_OPEN (SCRIPT_BODY | SCRIPT_SHORT_BODY) ;
style = STYLE_OPEN (STYLE_BODY | STYLE_SHORT_BODY) ;
htmlElement =  <TagOpen htmlElement TagClose> htmlElement
    | TagSingle
    | eps ;
\end{verbatim}
\end{center}

The following is a tagged CFG for XML adapted 
from ANTLR\footnote{\url{https://github.com/antlr/grammars-v4/tree/master/xml/}}.
\begin{center}
\begin{verbatim}
document  = prolog? misc* element misc*;
prolog    = XMLDeclOpen attribute* SPECIAL_CLOSE ;
content   = chardata? ((element | reference | CDATA | PI | COMMENT) chardata?)* ;
element   = OpenTag content CloseTag |   SingleTag ;
reference = EntityRef | CharRef ;
attribute = Name '=' STRING ;
chardata  = TEXT | SEA_WS ;
misc      = COMMENT | PI | SEA_WS ;
\end{verbatim}
\end{center}

\section{The translation algorithm} \label{app:terminate}

\newcommand{\transTable}{T}

\paragraph{Translating simple forms to linear forms}
We now give a description of the iterative translation algorithm.
The algorithm is based on the dependency graph $E_G$.
We start by removing from $E_G$ any dependency $(L,L')$ that is the
result of a rule like $L\ra s_1\br{L'}s_2$; intuitively, a matched
token $\br{L'}$ can be conceptually viewed as a ``plain symbol'' and
its presence does not affect the following translation.
Then, we remove any dependency $(L,L')$ that results from a rule like 
$L\ra sL'$. Based on the two conditions enforced by the validator,
the remaining $E_G$ does not have any cycles and becomes a directed acyclic graph (DAG).

In the following discussion, a rule whose head is $L$ is called \emph{a rule of $L$}. 

Our algorithm maintains 
a map $\transTable$ and a set $P'$ of translated rules.
$\transTable$ is initialized to an empty map and $P'$ is initialized to $P$. 
The algorithm uses $\transTable$ to keep track of new nonterminals created during translation.
At an iteration, a new nonterminal may be introduced for a certain string $L's$, 
where $L'\in V$ and $s \in (\Sigma\cup V)^+$. This new nonterminal is denoted as $L_{L's}$, and
the algorithm adds the mapping $(L_{L's}, L's)$ to $\transTable$.

At each iteration, if $E_G$ is empty, then the algorithm terminates (which implies
that all rules are already in linear forms).
Otherwise, since $E_G$ is a DAG, there must be a sink in $E_G$ (i.e., it does not have outgoing edges).
The algorithm selects a sink $L$ in $E_G$, and checks whether all rules of $L$ 
are in linear forms. If not, the algorithm rewrites non-linear-form rules of
$L$. There are three cases for such a rule:
\begin{enumerate}
    \item $L\ra Ls$, where $s\in (\Sigma\cup V)^*$. 
    This case cannot happen since the validator must reject the original grammar (it can be shown that
    the original grammar must have an invalid cycle rejected by the validator).
    
    \item $L\ra L's$, where $L'\neq L$ and $s\in (\Sigma\cup
    V)^*$. All rules of $L'$ must already be in linear forms; otherwise, $L'$ would
    be in $E_G$ and $L$ could not be a sink as $L$ depends on $L'$.
    The algorithm first uses the rules of $L'$ to rewrite $L \ra L's$
    into a set of new rules, then replaces $L \ra L's$ with the new rules in $P'$.
    The algorithm then updates the edges of $E_G$ based on the new rules.

    \item $L\ra t_1\cdots t_kL's$, where $t_k$ is a plain symbol
    or a matched token, and $s\in (\Sigma\cup V)^+$. First, the
    algorithm checks if there is a mapping $(L_{L's},L's)$
    in $\transTable$.  If not, the algorithm (1) creates a new
    nonterminal $L_{L's}$; (2) adds a new rule $L_{L's}\ra L's$ to $P'$;
    (3) adds a new mapping $(L_{L's},L's)$ to $\transTable$; and
    (4) adds a new node $L_{L's}$ to $E_G$.
    Second, the algorithm
    replaces $L\ra t_1\cdots t_kL's$ with $L\ra t_1\cdots t_kL_{L's}$ in $P'$
    and updates the edges of $E_G$ correspondingly.
\end{enumerate}

After all rules of $L$ are rewritten to be in linear forms, the
algorithm removes $L$ and the corresponding edges from $E_G$. After
that, the algorithm moves on to the next iteration.

Note that the updated $E_G$ after an iteration is still a DAG; so at
each iteration, a nonterminal can always be picked and progress can be
made.  Further, we can easily prove that the algorithm produces an
equivalent grammar as the original one by showing that each rewriting step
creates an equivalent grammar. Therefore, if the algorithm terminates
on some simple-form grammar that passes the validator, then it can be
translated to an equivalent linear-form grammar.  Note that we have
not proved that the algorithm always terminates (we have also not
found a counter example), which we leave for future work.

\section{An example of generating semantic actions}\label{app:gen_semanticactions}

\begin{equation*}
  \begin{split}
  & \text{Tagged CFG:} \\
  &L \ra A \br{A E}\ @L^6;\\
  &A \ra c E\ @A^2; \\
  &E \ra \epsilon\ @E^0.\\
  &
  \end{split}
  \Ra
  \begin{split}
  & \text{Simple forms:} \\
  &L \ra A \br{L_{AE}}\ @L^6;\\
  &L_{AE} \ra A E;\\
  &A \ra c E\ @A^2; \\
  &E \ra \epsilon\ @E^0.
  \end{split}
  \Ra 
  \begin{split}
  & \text{Linear forms:} \\
  &L  \ra c \br{ L_{AE} }E \ @L^6\circ A^1; \\
  &L_{AE} \ra c E\ @A^1; \\
  &A \ra c E\ @A^2; \\
  &E \ra \epsilon\ @E^0.
  \end{split}
  \Ra
  \begin{split}
  & \text{VPG:} \\
  &L  \ra c L_1 \ @L^6\circ A^1; \\
  &L_1  \ra \br{ L_{AE} }E; \\
  &L_{AE} \ra c E\ @A^1; \\
  &A \ra c E\ @A^2; \\
  &E \ra \epsilon\ @E^0.
  \end{split}
\end{equation*}

In the above tagged CFG, the first rule is attached with $L^6$ because an additional $E$ is implicitly added at the end of the rule.

The first step of translation separates $AE$ from the first rule and assigns it
to a new nonterminal $L_{AE}$. The semantic action of the first rule
does not change, since $L_{AE}$ has no semantic actions; so the
semantic values of $A$ and $E$ will be left on the stack. Thus,
$L^6$ still expects 6 values on the stack.
The second translation step expands $A$ in the first rule with $A\ra
cE @A^2$.  Then, $cE\br{ L_{AE} }E$ is simplified to $c\br{ L_{AE}}E$, 
and that is why $A^1$ is applied instead of $A^2$:
$A^1$ accepts the value for $c$.  {The same transformation is applied to the second rule.}
The last step of translation is more straightforward: a simple-form rule, e.g., $L\ra c_1c_2\dots c_nL'$,
is converted to $L\ra c_1L_1, L_1\ra c_2L_2,\cdots,L_{n-1}\ra c_n L'$.

As a concrete example, the parser will generate the following parse tree for the input string $c\br{c}$.
$$[(L,c,L_1),(L_1,\ac,L_{AE}),(L_{AE},c,E),((L_1,L_{AE}),\bc,E)].$$
Each edge in the above parse tree is then replaced with its attached action and the semantic values of the terminals.
{$$[L^6\circ A^1,c,\ac,A^1,c,E^0,\bc,E^0].$$}
And the evaluation result of the above stack machine is the following parse tree of the tagged CFG.
{$$[(L,[(A,[c]),\ac,(A,[c]),E^0,\bc,E^0])].$$}
\newpage
\section{The full evaluation}\label{app:eval}

The full evaluation is shown at 
Table~\ref{tab:evalJSON_full},
Table~\ref{tab:evalXML_full},
Table~\ref{tab:evalHTML_full},
Table~\ref{tab:total_JSON_full}, and
Table~\ref{tab:total_XML_full}.

\begin{table}[t]
    \centering
    \caption{The parsing time of JSON.}
    \label{tab:evalJSON_full}
    \begin{tabular}{crrrrr}
        \toprule
        Name	         & Size &   ANTLR     & VPG  & Conv  \\
        \midrule
        Members           & 74K & 17.30  ms & 1.468  ms & 0.534 ms \\
        poked             & 80K & 18.00  ms & 1.476  ms & 0.621 ms \\
        gists             & 89K & 14.53  ms & 0.684  ms & 0.257 ms \\
        senator           & 139K& 18.37  ms & 1.877  ms & 0.739 ms \\
        AskReddit         & 142K& 19.46  ms & 2.161  ms & 0.799 ms \\
        blog\_entries     & 148K& 10.68  ms & 0.551  ms & 0.065 ms \\
        github\_events    & 161K& 16.32  ms & 0.749  ms & 0.490 ms \\
        emojis            & 163K& 14.19  ms & 0.334  ms & 0.271 ms \\
        parliament\_events& 178K& 21.53  ms & 3.313  ms & 2.420 ms \\
        prize             & 214K& 21.71  ms & 2.433  ms & 2.443 ms \\
        y77d-th95         & 240K& 26.79  ms & 8.918  ms & 2.795 ms \\
        municipis         & 322K& 25.96  ms & 4.709  ms & 2.024 ms \\
        laureate          & 460K& 31.43  ms & 13.597 ms & 4.835 ms \\
        reddit\_all       & 480K& 26.27  ms & 6.128  ms & 3.169 ms \\
        transactions      & 530K& 30.40  ms & 5.749  ms & 7.256 ms \\
        representative    & 549K& 28.66  ms & 5.047  ms & 26.765 ms \\
        citm\_catalog     & 1.6M& 50.38  ms & 8.557  ms & 8.467 ms \\
        canada            & 2.1M& 92.99  ms & 14.898 ms & 20.273 ms \\
        twitter           & 2.1M& 23.09  ms & 6.612  ms & 2.884 ms \\
        movies            & 3.2M& 185.05 ms & 35.327 ms & 41.978 ms \\
        educativos        & 4.1M& 97.52  ms & 19.323 ms & 19.296 ms \\
        airlines          & 4.7M& 113.35 ms & 27.322 ms & 26.784 ms \\
        JSON.parse        & 7.0M& 234.66 ms & 62.955 ms & 61.066 ms \\
        \bottomrule
    \end{tabular}
\end{table}
\begin{table}[t]
    \centering
    \caption{The parsing time of XML.}
    \label{tab:evalXML_full}
    \begin{tabular}{crrrrr}
        \toprule
        Name	     & Size &	ANTLR    & VPG & Conv \\
        \midrule
        soap2       & 1.7K& 9.02 ms   & 0.0098 ms  & 0.01 ms \\
        nav\_48\_0  & 4.5K& 10.66 ms  & 0.0147 ms  & 0.02 ms \\
        nav\_63\_0  & 6.7K& 11.39 ms  & 0.0263 ms  & 0.03 ms \\
        nav\_78\_0  & 6.7K& 11.47 ms  & 0.0265 ms  & 0.03 ms \\
        cd\_catalog & 4.9K& 11.56 ms  & 0.0378 ms  & 0.065 ms \\
        form        & 15K & 12.13 ms  & 0.0494 ms  & 0.074 ms \\
        OfficeOrder & 10K & 12.37 ms  & 0.0538 ms  & 0.08 ms \\
        nav\_50\_0  & 10K & 12.51 ms  & 0.0621 ms  & 0.05 ms \\
        book        & 22K & 15.70 ms  & 0.1348 ms  & 0.264 ms \\
        book-order  & 22K & 15.88 ms  & 0.1343 ms  & 0.267 ms \\
        bioinfo     & 34K & 16.63 ms  & 0.3300 ms  & 0.301 ms \\
        soap\_small & 26K & 17.44 ms  & 0.1621 ms  & 0.326 ms \\
        cd\_big     & 30K & 18.35 ms  & 0.2271 ms  & 0.457 ms \\
        soap\_mid   & 131K& 26.39 ms  & 1.3509 ms  & 1.712 ms \\
        blog        & 1.3M& 38.13 ms  & 4.0320 ms  & 26.954 ms \\
        po1m        & 1.0M& 86.01 ms  & 4.2480 ms  & 12.795 ms \\
        soap        & 2.6M& 130.17 ms & 10.2351 ms & 34.011 ms \\
        bioinfo\_big& 4.3M& 146.46 ms & 17.0797 ms & 38.944 ms \\
        ORTCA       & 7.7M& 153.76 ms & 15.6059 ms & 6.007 ms \\
        SUAS        & 13M & 231.78 ms & 19.3162 ms & 13.343 ms \\
        address     & 15M & 429.40 ms & 67.7817 ms & 139.209 ms \\
        cd          & 26M & 912.88 ms & 192.2709 ms& 454.896 ms \\
        po          & 73M & 2058.21 ms& 425.1013 ms& 1070.37 ms \\
        \bottomrule
    \end{tabular}
\end{table}
\begin{table}[t]
    \centering
    \caption{The parsingtime of HTML.}
    \label{tab:evalHTML_full}
    \begin{tabular}{crrrr}
        \toprule
        Name     & Size & ANTLR   & VPG  & Conv   \\
        \midrule
        uglylink  & 172B& 15.76   ms & 0.0101 ms &  0.000002 ms \\
        style1    & 195B& 20.75   ms & 0.0146 ms &  0.000004 ms \\
        script1   & 277B& 20.54   ms & 0.0142 ms &  0.000005 ms \\
        attvalues & 384B& 26.44   ms & 0.0155 ms &  0.000004 ms \\
        html4     & 750B& 25.05   ms & 0.0171 ms &  0.000005 ms \\
        antlr     & 9.3K& 1322.56 ms & 0.0683 ms &  0.052 ms \\
        gnu       & 21K & 1992.10 ms & 0.0948 ms &  0.1 ms \\
        freebsd   & 27K & 3883.07 ms & 0.1026 ms &  0.146 ms \\
        abc.com   & 50K & 3839.71 ms & 0.1181 ms &  0.16 ms \\
        github    & 51K & 9818.   ms & 0.1806 ms &  0.261 ms \\
        metafilter& 63K & 11932.0 ms & 0.1739 ms &  0.275 ms \\
        wikipedia & 67K & 13005.6 ms & 0.3106 ms &  0.296 ms \\
        nbc.com   & 95K & 22154.1 ms & 0.2277 ms &  0.409 ms \\
        bbc       & 110K& 18913.9 ms & 0.2254 ms &  0.402 ms \\
        reddit    & 114K& 53793.7 ms & 0.2999 ms &  0.576 ms \\
        reddit2   & 114K& 51580.4 ms & 0.2955 ms &  0.571 ms \\
        cnn1      & 118K& 36689.5 ms & 0.5208 ms &  0.522 ms \\
        google    & 143K& 1832.42 ms & 0.0785 ms &  0.072 ms \\
        digg      & 152K& 59766.  ms & 0.6093 ms &  0.669 ms \\
        youtube   & 489K& 542537. ms & 1.7605 ms &  2.285 ms \\
        \bottomrule
    \end{tabular}
\end{table}
    
\begin{table}[t]
  \centering
  \caption{The total time (ms) of parsing JSON.}
  \label{tab:total_JSON_full}
  \begin{tabular}{crrrrrrrr}
      \toprule
    Name        & Size& \begin{tabular}{@{}c@{}}ANTLR \\ Lex\end{tabular} & \begin{tabular}{@{}c@{}}VPG \\ Parse\end{tabular}& Lex+Parse& SpiderM & JSCore& V8  & Chakra \\
      \midrule
Members       & 74K & 12.81 & 1.46 & 14.28 & 14.51 & 51.9  & 8.85 & 13.21 \\
poked         & 80K & 14.27 & 1.47 & 15.75 & 14.67 & 51.59 & 9.28 & 13.42 \\
gists         & 89K & 12.08 & 0.68 & 12.77 & 14.82 & 50.83 & 8.72 & 12.68 \\
senator       & 139K& 16.79 & 1.87 & 18.67 & 15.78 & 66.17 & 9.99 & 12.43 \\
AskReddit     & 142K& 16.3  & 2.16 & 18.48 & 15.4  & 46.53 & 9.19 & 14.14 \\
blog\_entries & 148K& 10.72 & 0.55 & 11.27 & 15.06 & 45.33 & 8.8  & 13.05 \\
github\_events& 161K& 16.57 & 0.74 & 17.32 & 16.01 & 55.43 & 9.98 & 13.87 \\
emojis        & 163K& 12.43 & 0.33 & 12.76 & 16.29 & 85.23 & 9.21 & 13.72 \\
parliament    & 178K& 16.30 & 3.31 & 19.62 & 16.36 & 70.35 & 10.38& 13.69 \\
prize         & 214K& 16.81 & 2.43 & 19.24 & 16.76 & 134.3 & 9.98 & 13.15 \\
y77d-th95     & 240K& 22.76 & 8.91 & 31.68 & 17.41 & 105.22& 11.96& 15.37 \\
municipis     & 322K& 23.58 & 4.70 & 28.29 & 20.72 & 178.03& 12.53& 16.35 \\
laureate      & 460K& 24.86 & 13.59& 38.46 & 20.72 & 178.03& 12.53& 16.35 \\
reddit\_all   & 480K& 25.51 & 6.12 & 31.64 & 20.02 & 60.2  & 11.48& 16.54 \\
transactions  & 530K& 24.18 & 5.74 & 29.93 & 21.09 & 78.9  & 11.78& 18.35 \\
representative& 549K& 29.26 & 5.04 & 34.31 & 22.27 & 132.86& 15.21& 16.19 \\
citm\_catalog & 1.6M& 38.8  & 8.55 & 47.38 & 33.81 & 70.64 & 27.59& 24.77 \\
canada        & 2.1M& 45.1  & 14.89& 60.09 & 57.14 & 67.51 & 33.95& 43.83 \\
twitter       & 2.1M& 103.7 & 6.61 & 110.36& 22    & 92.54 & 15.61& 17.44 \\
movies        & 3.2M& 85.39 & 35.32& 120.72& 70.93 & 106.41& 68   & 50.85 \\
educativos    & 4.1M& 108.14& 19.32& 127.47& 70.78 & 420.99& 44.53& 48.64 \\
airlines      & 4.7M& 80.97 & 27.32& 108.30& 73.86 & 95.12 & 41.93& 55.52 \\
Google        & 7.0M& 129.9 & 62.95& 192.93& 118.11& 139.21& 76.36& 87.81 \\
\bottomrule
\end{tabular}
\end{table}

\begin{table}[t]
  \caption{The total time (ms) of parsing XML.}
    \label{tab:total_XML_full}
  \begin{tabular}{crrrrrrrr}
  \toprule
  Name        & Size& \begin{tabular}{@{}c@{}}ANTLR \\ Lex\end{tabular}       &  \begin{tabular}{@{}c@{}}VPG \\ Parse\end{tabular}   & Lex+Parse     & Fast-XML& Libxmljs& SAX-JS  & HP2 \\
  \midrule
soap2       & 1.7K& 8.09  & 0.010  & 8.101  & 0.057   & 0.051  & 0.142   & 0.049           \\
nav\_48\_0  & 4.5K& 8.58  & 0.015  & 8.597  & 0.098   & 0.097  & 0.496   & 0.163           \\
cd\_catalog & 4.9K& 9.18  & 0.026  & 9.211  & 0.223   & 0.093  & 0.363   & 0.158           \\
nav\_63\_0  & 6.7K& 9.26  & 0.027  & 9.289  & 0.143   & 0.126  & 0.671   & 0.198           \\
nav\_78\_0  & 6.7K& 9.20  & 0.038  & 9.239  & 0.124   & 0.125  & 0.682   & 0.199           \\
nav\_50\_0  & 10K & 9.92  & 0.049  & 9.973  & 0.213   & 0.174  & 0.995   & 0.303           \\
OfficeOrder & 10K & 10.7 & 0.054  & 10.8 & 0.369   & 0.137  & 0.881   & 0.346           \\
form        & 15K & 11.6 & 0.062  & 11.7 & 0.375   & 0.168  & 0.934   & 0.272           \\
book-order  & 22K & 11.6 & 0.135  & 11.8 & 0.693   & 0.260  & 1.413   & 0.536           \\
book        & 22K & 11.8 & 0.134  & 11.9 & 0.635   & 0.260  & 1.410   & 0.544           \\
soap\_small & 26K & 12.3 & 0.330  & 12.6 & 0.855   & 0.402  & 2.421   & 0.776           \\
cd\_big     & 30K & 13.1 & 0.162  & 13.3 & 1.171   & 0.436  & 2.199   & 0.877           \\
bioinfo     & 34K & 15.6 & 0.227  & 15.9 & 0.992   & 0.379  & 2.157   & 0.901           \\
soap\_mid   & 131K& 21.2 & 1.351  & 22.6 & 4.109   & 1.953  & 11.7  & 3.813           \\
po1m        & 1.0M& 46.1 & 4.032  & 50.1 & 33.509  & 11.558 & 95.2  & 27.1          \\
blog        & 1.3M& 52.8 & 4.248  & 57.0 & 29.575  & 6.560  & 48.0  & 17.8          \\
soap        & 2.6M& 99.1 & 10.2 & 109.3& 110.1 & 38.1 & 260.8 & 76.2          \\
bioinfo\_big& 4.3M& 122.5& 17.0 & 139.6& 116.9 & 42.7 & 276.9 & 82.2          \\
ORTCA       & 7.7M& 141.8& 15.6 & 157.4& 138.4 & 91.2 & 664.6 & 88.7          \\
SUAS        & 13M & 236.5& 19.3 & 255.8& 254.3 & 182.4& 1213.6& 168.7         \\
address     & 15M & 367.2& 67.7 & 435.0& 584.4 & 196.1& 1012.0& 330.5         \\
cd          & 26M & 731.7& 192.2& 924.0& 1298& 419.0& 2102& 734.9         \\
po          & 73M & 1812 & 425  & 2237 & 3278  & 897  & 6617  & 1827        \\
  \bottomrule
  \end{tabular}
\end{table}

\section{Profiler result for the HTML grammar}\label{app:profile}
Table \ref{tab:bbc.profile} shows the profiler result of ANLTR for parsing ``bbc.com.html''.
In the table,
``Invocations'' means the number of decision invocations, 
``Time'' means the estimate time for prediction,
``Total k'' means the total number of lookahead symbols examined,
``Max k'' means the maximal number of lookahead symbols examined in any decision event,
``Ambiguity'' means the number of ambiguous input phrases,
and ``DFA cache miss'' means the number of non-DFA transitions during prediction.
The main time cost in parsing the HTML file (as well as other HTML files) is from 
the prediction and DFA cache miss.

\begin{table}[t]
    \centering
    \caption{The profiler result of ANTLR for parsing ``bbc.com.html''. The file triggers many DFA cache miss, which seriously slows down the parser.}
    \label{tab:bbc.profile}
    \begin{tabular}{l|l|l|l|l|l}
        \toprule
        Invocations & Time & Total k & Max k & Ambiguous & DFA cache miss \\
        \midrule
        2 & 0.177 & 5 & 2 & 1 & 5 \\
        1 & 0.005 & 1 & 1 & 0 & 1 \\
        1 & 0.005 & 1 & 1 & 0 & 1 \\
        1 & 0.004 & 1 & 1 & 0 & 1 \\
        2 & 0.111 & 15 & 12 & 1 & 13 \\
        2 & 0.004 & 2 & 1 & 0 & 2 \\
        1 & 0.004 & 1 & 1 & 0 & 1 \\
        2 & 0.073 & 3 & 2 & 0 & 3 \\
        2478 & 0.451 & 2478 & 1 & 0 & 3 \\
        938 & 24867.858 & 6389021 & 14687 & 0 & 5778095 \\
        1003 & 0.303 & 1003 & 1 & 0 & 2 \\
        1087 & 0.26 & 1087 & 1 & 0 & 3 \\
        938 & 0.525 & 938 & 1 & 0 & 3 \\
        1127 & 0.295 & 1127 & 1 & 0 & 4 \\
        1127 & 0.17 & 1127 & 1 & 0 & 3 \\
        2065 & 0.695 & 4005 & 2 & 0 & 6 \\
        1475 & 0.302 & 1475 & 1 & 0 & 1 \\
        1 & 0.002 & 1 & 1 & 0 & 1 \\
        \bottomrule
    \end{tabular}
\end{table}

\end{document}